\documentclass[journal,twocolumn]{IEEEtran}
\IEEEoverridecommandlockouts
\usepackage{cite}
\usepackage{amsmath,amssymb,amsfonts,amsthm}
\usepackage{algorithmic}
\usepackage{graphicx}
\usepackage{textcomp}
\usepackage{xcolor}
\usepackage{mathtools}
\usepackage[short]{optidef}

\newtheorem{lemma}{Lemma}
\newtheorem{remark}{Remark}

\DeclareMathOperator*{\argmin}{argmin}

\usepackage[hidelinks,citecolor={blue}]{hyperref}

\def\BibTeX{{\rm B\kern-.05em{\sc i\kern-.025em b}\kern-.08em
    T\kern-.1667em\lower.7ex\hbox{E}\kern-.125emX}}
\begin{document}
\title{Liquid Lens-Based Imaging Receiver for MIMO VLC Systems\\}

\author{Kapila~W.~S.~Palitharathna,~\IEEEmembership{Member,~IEEE,} Christodoulos~Skouroumounis,~\IEEEmembership{Senior~Member,~IEEE,} and
Ioannis~Krikidis,~\IEEEmembership{Fellow,~IEEE}
\thanks{The authors are with the IRIDA Research Centre for Communication Technologies, Department of Electrical and Computer Engineering, University of Cyprus, 1678 Nicosia, Cyprus (e-mails: \{palitharathna.kapila, cskour03, krikidis\}@ucy.ac.cy).}\vspace{-10mm}
\thanks{This work received funding from the European Research Council (ERC) under the European Union's Horizon 2020 research and innovation programme (Grant agreement No. 819819) and from the Smart Networks and Services Joint Undertaking (SNS JU) under the European Union's Horizon Europe research and innovation programme (Grant Agreement No 101192080). It was also funded by the European Union Marie Skłodowska-Curie Actions Project COALESCE under Grant 101130739.}}

\maketitle

\begin{abstract}
In this paper, we consider a tunable liquid convex lens-assisted imaging receiver for indoor multiple-input multiple-output (MIMO) visible light communication (VLC) systems. In contrast to existing MIMO VLC receivers that rely on fixed optical lenses, the proposed receiver leverages the additional degrees of freedom offered by liquid lenses via adjusting both focal length and orientation angles of the lens. This capability facilitates the mitigation of spatial correlation between the channel gains, thereby enhancing the overall signal quality and leading to improved bit-error rate (BER) performance. We present an accurate channel model for the liquid lens-assisted VLC system by using three-dimensional geometry and geometric optics. To achieve optimal performance under practical conditions such as random receiver orientation and user mobility, optimization of both focal length and orientation angles of the lens are required. To this end, driven by the fact that channel models are mathematically complex, we present two optimization schemes including a blockwise machine learning (ML) architecture that includes convolution layers to extract spatial features from the received signal, long-short term memory layers to predict the user position and orientation, and fully connected layers to estimate the optimal lens parameters. Numerical results are presented to compare the performance of each scheme with conventional receivers. Results show that a significant BER improvement is achieved when liquid lenses and presented ML-based optimization approaches are used. Specifically, the BER can be improved from $6\times 10^{-2}$ to $1.4\times 10^{-3}$ at an average signal-to-noise ratio of $30$ dB.
\end{abstract}


\vspace{-5mm}
\section{Introduction}
Next-generation wireless networks are expected to provide exceptionally high data rates, ultra-low latency, and enhanced security, surpassing the performance capabilities of current wireless systems. In this context, visible light communication (VLC) is positioned to play a pivotal role in future wireless networks, serving as a promising technology that offers high data rates, extremely low latency, and improved security, especially for indoor short-range communication systems. Compared to radio frequency (RF) communication, VLC can achieve data rates of several terabits per second in indoor environments, leveraging the extensive wavelength range of the visible light spectrum, which spans from 380 nm to 780 nm~\cite{Ghassemlooy}. Moreover, VLC systems are capable of simultaneously providing both energy-efficient lighting and high-speed indoor communication by using low-cost light-emitting diodes (LEDs) as transmitters, while photodiodes (PDs) are used as receivers.

Another significant challenge that prominently emerges within the landscape of future wireless networks is meeting the growing demand for seamless communication in an increasingly interconnected world. In response to this evolving landscape, extensive research efforts have been dedicated to advancing multi-user/multi-antenna communication technologies. Often referred to as multiple-input multiple-output (MIMO) systems, these technologies are designed to elevate spectral efficiency through employing sophisticated spatial multiplexing techniques \cite{Renzo_2014}. In the context of VLC networks, MIMO has emerged as a promising solution to enhance data rates and improve link reliability by leveraging the inherent high directivity of LED light patterns, coupled with the strategic spatial placement of LEDs and PDs. In particular, the inter-LED distance, LED placement geometry (\textit{e.g.,} rectangular, circular, grid like), and the receiver geometry (\textit{e.g.,} planer, hemispherical, pyramid) significantly influence the performance~\cite{chockalingam_2015, Asanka_2015}. In the realm of MIMO VLC systems, the spatial modulation (SM) technique has garnered significant attention due to its robust performance in highly correlated channel conditions, especially when compared to the spatial multiplexing (SMP) technique~\cite{Fath_2023}. As a generalized form of SM, generalized spatial modulation (GSM) offers superior transmission efficiency, making it a highly suitable candidate for MIMO VLC applications~\cite{chockalingam_2015}. Research has demonstrated that GSM outperforms other modulation schemes in terms of bit-error rate (BER) within the VLC domain~\cite{chockalingam_2015}. Various studies have since focused on enhancing the performance of GSM-based VLC systems~\cite{Nahhal_2021, Wang_2018, Chen_2021}. For instance, the work presented in~\cite{Nahhal_2021} introduces a flexible GSM scheme that dynamically adjusts modulation sizes across the LEDs and the number of active LEDs, thereby improving the average symbol error rate and spectral efficiency. Additionally, a spectral-efficient GSM-based hybrid dimming scheme utilizing layered asymmetrically clipped optical orthogonal frequency division multiplexing is proposed in~\cite{Wang_2018}, which integrates both spatial-domain and time-domain dimming strategies. Furthermore, in~\cite{Chen_2021}, a group-based LED selection mechanism is proposed to enhance link robustness by selectively excluding poorly performing LEDs from the communication process.

However, achieving full diversity and spatial multiplexing gains in MIMO systems requires maintaining a minimum antenna separation of at least half the operating wavelength, which presents a significant challenge for mobile devices constrained by physical space limitations. In the context of MIMO VLC systems, channel correlation is influenced by several factors, including the spatial positioning of LEDs and PDs, inter-element spacing, the radiation pattern of the LEDs, and the field-of-view (FoV) of the PDs~\cite{chockalingam_2015}. The inherently low inter-LED and inter-PD spacing in practical VLC transmitters and receivers, along with the use of planar transmitter/receiver arrays, are primary contributors to the pronounced degradation in BER performance. To mitigate correlation among the elements of the channel matrix, various techniques have been explored, including optimizing transmitter and receiver geometries, incorporating intelligent reflecting surfaces (IRSs), and employing optical lenses. These approaches have demonstrated significant improvements in reducing channel correlation and enhancing overall system performance~\cite{Asanka_2015,Shiyuan_2023,Shiyuan_2024,Sushanth_2018, Wang_2013, Jiang_2018}. In~\cite{Asanka_2015}, an angle diversity receiver utilizing pyramid and hemispherical geometries is proposed to enhance signal reception capabilities in VLC systems. The integration of IRSs within MIMO VLC frameworks is systematically analyzed in~\cite{Shiyuan_2023} and~\cite{Shiyuan_2024}, highlighting their potential to improve spatial multiplexing performance. Additionally, the adoption of an imaging receiver, equipped with an imaging lens, effectively mitigates performance degradation by minimizing channel correlation, as the lens facilitates the focusing of light beams emitted from LEDs onto distinct PDs \cite{Sushanth_2018}. The performance attributes of various lens types, including convex lenses, hemispherical lenses, and fish-eye lenses, have been rigorously examined in the literature, underscoring their respective contributions to optimizing the efficiency and reliability of VLC systems \cite{Sushanth_2018, Wang_2013, Jiang_2018}.

Recent advancements in optical lens technology have introduced adaptable and tunable liquid lens architectures, which offer promising potential for significantly enhancing communication efficiency by dynamically optimizing optical paths and improving signal quality~\cite{ABoagye_2023, Ndjiongue_2021, Ngatched_2021, Cheng_2021, Lee_2019, Zohrabi_2016, Tian_2022}. Among these innovations, liquid crystal-based structures that adjust the refractive index to manipulate light propagation direction have been thoroughly examined within the field of VLC systems~\cite{ABoagye_2023, Ndjiongue_2021}. In~\cite{ABoagye_2023}, a liquid crystal-based IRS is employed as a VLC transmitter to enhance data rate uniformity among users. Additionally, in~\cite{Ndjiongue_2021} and~\cite{Ngatched_2021}, a liquid crystal-based IRS is utilized at the receiver to dynamically steer light beams toward the effective area of the PD, thereby optimizing signal reception and improving overall communication performance. Although omitted in the context of VLC systems, several new non-mechanical liquid lens architectures have been proposed that can change the orientation and shape of the liquid surface and hence can control the light propagation direction~\cite{Cheng_2021, Lee_2019, Zohrabi_2016}. Numerous non-mechanical electro-wetting surface-based liquid lens architectures have been proposed in the literature, such as those detailed in~\cite{Cheng_2021}. The authors in~\cite{Lee_2019} introduced an innovative three-dimensional beam steering methodology that leverages an electro-wetting-based liquid lens in conjunction with a liquid prism, enhancing the precision of light manipulation.  In~\cite{Zohrabi_2016}, the authors presented techniques for one- and two-dimensional beam steering employing multiple tunable liquid lenses, demonstrating significant advancements in beam control capabilities. Furthermore, recent research has explored the integration of liquid lens systems with mechanical structures to achieve enhanced dynamic beam steering functionalities~\cite{Tian_2022,Zhang_2025,Lv_2025}. In~\cite{Tian_2022}, an adaptable liquid lens is studied which has three degrees of freedom \textit{i.e.}, focal length, azimuth angle, and polar angle. The adjustment of focal length is facilitated by the application of a vertical mechanical force on a ring positioned around the liquid, while a mechanical framework employs magnetic forces to enable tilting of the ring, allowing for precise modulation of both azimuth and polar angles~\cite{Tian_2022}. \textcolor{black}{Moreover, the study in~\cite{Zhang_2025} demonstrates that the tuning capability of liquid lenses operates within a timescale of milliseconds and is practically feasible for VLC systems. In~\cite{Lv_2025}, a bionic optical focusable imaging system that mimics the focusing mechanism of the human eye has been proposed, featuring fast tuning capability.} It is identified that a combination of such liquid lenses with imaging receivers and optimization of lens parameters can provide a robust solution for MIMO VLC systems under dynamic conditions. In particular, by adjusting the focal length and orientation angles of the lens, the MIMO receivers can focus light beams coming from LEDs on to separated PDs such that interference is minimized and can achieve full multiplexing gains.

User mobility and random receiver orientations resulting from human activities, such as sitting and walking, significantly influence the performance of VLC systems, as these movements induce temporal variations in channel gains. Designing a receiver that maintains robust communication under these dynamic conditions presents considerable challenges. Furthermore, the ability to predict system parameters for future time instances and to optimize accordingly is critical, as it enables systems to mitigate processing delays and enhance overall communication efficiency. However, in practical indoor environments, user mobility and random receiver orientations often deviate from the assumptions underlying existing mathematical models. Consequently, traditional methods for estimating user positions and receiver orientations, along with parameter optimization schemes derived from theoretical models, may yield suboptimal performance gains. This discrepancy underscores the necessity for adaptive approaches that account for real-world dynamics in order to enhance the effectiveness of VLC systems. To address these challenges, data-driven machine learning (ML)-based solutions have emerged as promising methodologies~\cite{Arfaoui_2021, Kapila_2024, Kapila_2023}. In~\cite{Arfaoui_2021}, ML-based precise 3D positioning and orientation estimation techniques for VLC have been proposed. In addition, intelligent systems that can predict the user position and orientation for future time instance and accordingly optimize system parameters have received the attention of researchers~\cite{Kapila_2024, Kapila_2023}. In~\cite{Kapila_2024}, an ML-based user path and orientation prediction scheme has been proposed to optimize system parameters of an indoor VLC system. In~\cite{Kapila_2023}, average received power levels were used to predict the user blockage and accordingly optimized the beamforming matrix of VLC/RF hybrid systems. However, to the best of the authors' knowledge, this is the first study that uses liquid lenses for MIMO VLC systems and optimizes lens parameters using a prediction-based ML architecture to achieve better performance under user mobility and random receiver orientations.

In this paper, we investigate an indoor MIMO VLC system that incorporates a liquid convex lens-assisted imaging receiver, by utilizing GSM to facilitate information transmission. The use of GSM in our system is motivated by its superior BER performance in VLC systems compared to other MIMO schemes such as SMP, space shift keying (SSK), generalized SSK, and SM, as well as its high transmission efficiency~\cite{chockalingam_2015}. Additionally, GSM is a generalized scheme, with SSK and SMP serving as special cases of GSM. In contrast to the static conditions commonly examined in existing literature, this work focuses on dynamic environments characterized by user mobility and unpredictable receiver orientations. We propose the integration of an adjustable liquid convex lens at the receiver to enhance system performance and ensure superior coverage across a wide range of dynamic conditions. A comprehensive analysis of the channel gain under these practical scenarios is provided. To optimize system performance, we adjust the focal length and orientation angles of the liquid lens, employing two optimization schemes: a prediction-based block-wise ML architecture and a nearest LED selection approach. Additionally, we introduce two baseline techniques—exhaustive search and a vertically upward lens configuration—for comparative analysis. The ML architecture encompasses multiple components, including position and orientation estimation, position and orientation prediction, and optimization of lens parameters, demonstrating near-optimal results. Our contributions are summarized as follows:
\vspace{-3mm}
\begin{itemize}
    \item We propose a liquid convex lens model designed to enhance channel gains in the MIMO VLC system, while effectively minimizing channel gain correlation across a broad spectrum of user mobility and random receiver orientation scenarios. Furthermore, we establish a robust mathematical framework to accurately characterize the individual channel gains of the proposed system, incorporating key parameters such as user position, receiver orientation, lens focal length, and lens orientation. 
    \item We formulate an optimization problem to minimize the BER of the liquid lens-assisted MIMO VLC system by adjusting the focal length and the orientation angles of the liquid lens. In particular, the lens parameters are constrained by hardware limitations of the lens. To solve the formed problem, two optimization schemes including a prediction-based block-wise ML architecture (PBML) and the closest LED selection (CLS) are proposed. The PBML architecture includes a position and receiver orientation estimation, position and receiver orientation prediction, and parameter optimization blocks. Our ML architecture can be trained block-wise and is capable of predicting the position and orientation of the receiver and optimizing system parameters accordingly to achieve near optimal results.
    \item Results reveal that the use of liquid lenses in MIMO VLC systems is helpful to improve the BER performance under user mobility and random receiver orientation conditions. Moreover, the presented PBML architecture is capable of achieving near-optimal results and has a significant performance improvement compared to other schemes. Our solution can obtain optimal lens parameters in a low time duration and has shown robustness under a wide range of dynamic conditions. Also, the position and orientation prediction technique helps to apply the solution in many practical conditions where optimization for a future time instance is required.
\end{itemize}
The remainder of this article is organized as follows. In Section \ref{sec:system}, the system model and the associated liquid lens architecture, user mobility and random receiver orientation models are presented. Section \ref{sec:channel} characterizes the MIMO VLC channel and presents useful expressions to analyze the channel gains under dynamic conditions. In Section \ref{sec:BER}, the optimization problem to minimize the BER by adjusting lens parameters is presented and optimization schemes including an PBML architecture to obtain optimal lens parameters are also presented. Numerical results for various system parameters are presented in Section \ref{sec:results}. Finally, conclusions are drawn in Section \ref{sec:conclusion}.
\begin{figure}[!t]
    \centering
    \includegraphics[width=0.97\columnwidth]{"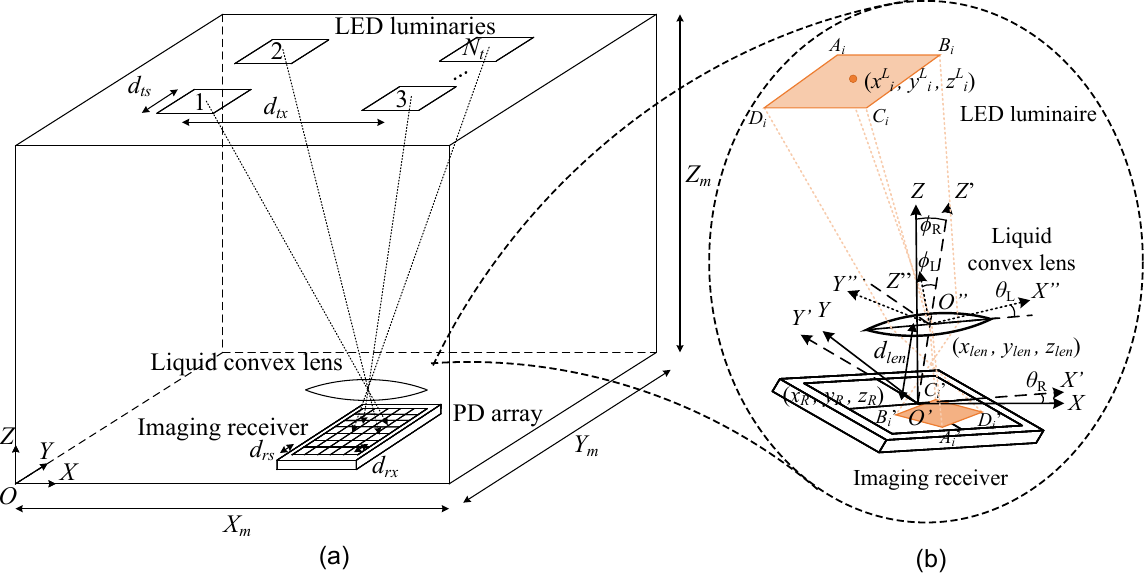"}
    \vspace{-3mm}
    \caption{(a) Liquid lens-assisted imaging receiver-based MIMO VLC system. (b) Liquid lens-assisted imaging receiver.}
    \label{fig:f1}
    \vspace{-6mm}
\end{figure}

\vspace{-6mm}
\section{System Model}\label{sec:system}
In this section, we provide the details of the considered MIMO VLC system model, the adopted liquid lens architecture, the employed MIMO technique, the channel model, the mobility and the receiver orientation models.

\vspace{-5mm}
\subsection{Network Topology}
We consider an indoor MIMO VLC system consisting of a ceiling mounted $N_t$ number of square LED luminaries, each with side length $d_{ts}$ and inter-LED distance of $d_{tx}$ that illuminate the room of size $X_m \text{m}\times Y_m \text{m}\times Z_m \text{m}$ and transmit data to a mobile device as shown in Fig. \ref{fig:f1}(a). Let $OXYZ$ be the room's coordinate frame. The coordinates of the geometric center of the $i$-th LED is $\hat{\boldsymbol{P}}_{i} = (x_i^L,y_i^L,z_i^L)$ with respect to (w.r.t.) $OXYZ$ frame. The mobile device is equipped with an imaging receiver which includes $N_r$ number of square PDs, each with side length $d_{rs}$ placed on a plane with inter-PD distance $d_{rx}$~\cite{Sushanth_2018}. The coordinates of the center of the PD plane are $\hat{\boldsymbol{P}}_{R} =(x_R,y_R,z_R)$ w.r.t. $OXYZ$ frame. We denote the receiver's coordinate frame as $O'X'Y'Z'$. In order to focus light beams on the PDs and to minimize interference, a reconfigurable convex lens~\cite{Sushanth_2018} is placed on the receiver such that its centroid is fixed at a distance of $d_{len}$ on the $z$-axis of the $O'X'Y'Z'$ frame as shown in Fig. \ref{fig:f1}(b). In this work, we consider the random receiver orientation and user mobility, and hence, the receiver's coordinate frame, $O'X'Y'Z'$, can be subjected to an azimuth angle of $\theta_R$ and polar angle of $\phi_R$ w.r.t. $OXYZ$ frame\footnote{The size of the PD array is selected such that the formed light spots can be focused on to the PD array for all the user positions and orientations.}. 

\subsection{{Liquid Lens Architecture}}
The proposed liquid lens architecture used in this paper is shown in Fig. \ref{fig:f3_lens}, in which a chamber is filled with a liquid and by changing the magnetic force applied on an annular metal ring on the surface, the orientation and the shape of the convex surface can be  adjusted~\cite{Tian_2022}. In the proposed liquid lens architecture, the focal length of the lens, $f$, can be adjusted by varying the vertical force applied on the annular ring~\cite{Tian_2022}. Moreover, the azimuth angle, $\theta_L$ and polar angle, $\phi_L$ of the lens can be adjusted w.r.t. $O'X'Y'Z'$ frame by varying the force applied on the magnet and the driving ring to focus light spots from the LEDs onto the imaging receiver, and hence, to improve the system's performance~\cite{Tian_2022}\footnote{\textcolor{black}{The liquid lens utilizes materials and fluids with low friction, enabling adjustments within several milliseconds~\cite{Zhang_2025}, which is shorter than the coherence time of a indoor VLC setup~\cite{soltani_2019}. Investigating the transient effects during lens adjustment represents an interesting direction for future research.}}. Further, we denote a coordinate frame in the center of the lens as $O''X''Y''Z''$ that will be useful in deriving channel gains. 
\begin{figure}[!t]
    \centering
    \includegraphics[width=0.96\columnwidth]{"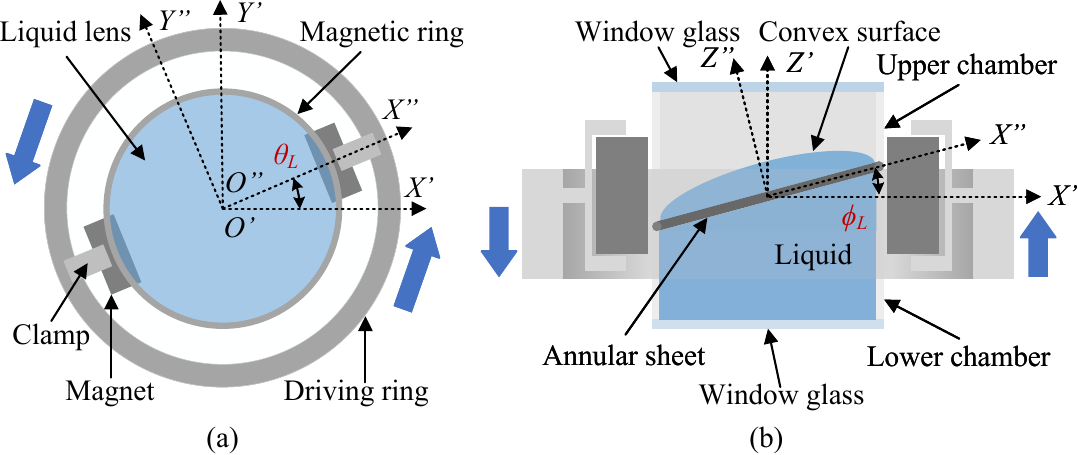"}
        \vspace{-4mm}
    \caption{Liquid lens structure and operation mechanism. (a). Top view. (b). Side view.}
    \label{fig:f3_lens}
    \vspace{-5mm}
\end{figure}

\vspace{-3mm}
\subsection{GSM Scheme}
We use GSM to convey information to the MIMO receiver. In GSM, bits are conveyed through modulation symbols sent on active LEDs as well as through LED activation patterns~\cite{chockalingam_2015}. In a single channel use, $N_a$ out of $N_t$ LEDs are selected to be activated, and each active LED emits an $M$-ary intensity modulated symbol from the set of intensity levels $\mathbb{M}$ where the $m$-th intensity level, $I_m$, can be expressed as~\cite{chockalingam_2015}
\begin{equation}
    I_m = \frac{2I_Pm}{M+1}, \quad \quad m = 1,2,\ldots, M,
    \label{equ:e4}
\end{equation}
where $M = |\mathbb{M}|$, and $I_P$ is the mean optical power emitted. Hence, the total number of bits conveyed per channel use {(bpcu)} is~\cite{chockalingam_2015}
\vspace{-3mm}
\begin{equation}
    \eta_{gsm} = \left\lfloor \log_2\left(\begin{array}{c}
        N_t  \\
        N_a
    \end{array}\right)\right\rfloor +N_a\lfloor\log_2 M\rfloor \quad \text{bpcu}.
    \label{equ:e5}
\end{equation}
The transmit signal vector of dimension $N_t\times 1$ is $\mathbf{x} = [x_1 \quad x_2 \quad \cdots \quad x_{N_t}]^T$, where $x_i\in\{\mathbb{M}\cup0\}$ is the transmit signal at the $i$-th luminary. The $N_r\times 1$ received signal vector $\mathbf{y}$ in the electrical domain at the receiver can be expressed as

\begin{equation}
    \mathbf{y} = \alpha r\mathbf{H}\mathbf{x}+\mathbf{n},
    \label{equ:e1}
\end{equation}
where $\alpha$ is the electrical-to-optical conversion efficiency of LEDs, and $r$ is the responsivity of PDs. $\mathbf{H}$ is the $N_r \times N_t$ dimentional optical channel gain matrix whose $(j,i)$-th element, $h_{i,j}$, is the optical channel gain from the $i$-th LED luminary to the $j$-th PD, and $\mathbf{n}=[n_1 \quad n_2 \quad \cdots \quad n_{N_r}]^T$ is the noise vector, where each element $n_j$ is real AWGN noise with zero mean and variance $\sigma^2$. 
The average signal-to-noise ratio (SNR) can be expressed as
\vspace{-3mm}
\begin{equation}
    \bar{\gamma}=\frac{\alpha^2r^2}{\sigma^2N_r}\sum_{j=1}^{N_r}\mathbb{E}\{(\mathbf{h}_j\mathbf{x})^2\},
    \label{equ:e2}
\end{equation}
where $\mathbf{h}_j$ is the $j$-th row of the matrix $\mathbf{H}$, and $\mathbb{E}\{\cdot\}$ is the expectation operator, respectively.

We use maximum likelihood detection at the receiver. Let $\mathbb{S}_{Tx} = \{\mathbf{x}_1, \mathbf{x}_2, \cdots, \mathbf{x}_L\}$ be the set of all possible transmit signal vectors for a given GSM scheme. The maximum likelihood detection rule for VLC MIMO systems can be expressed as 
\vspace{-1mm}
\begin{equation}
    \tilde{\mathbf{x}}=\argmin_{\mathbf{x}\in \mathbb{S}_{Tx}} ||\mathbf{y}-\alpha r \mathbf{H}\mathbf{x}||^2.
    \label{equ:e3}
\end{equation}

\vspace{-5mm}
\subsection{{Channel Model}}
In this subsection, we explain the light propagation model used in our setup. The optical channel gain $h_{i,j}$ from the $i$-th LED to the $j$-th PD can be expressed as~\cite{Sushanth_2018}

\begin{equation}
    h_{i,j} = h_{i}^{LoS}h_{i,j}^{len},
    \label{equ:e4}
\end{equation}
where $h_{i}^{LoS}$ is the line-of-sight (LoS) light propagation model from the $i$-th LED luminaire to the aperture of the imaging lens, and $h_{i,j}^{len}$ is the imaging channel gain between the $i$-th luminaire and the $j$-th PD. The LoS light propagation in an indoor VLC system is deterministic in nature and the well-known Lambertian model is used to obtain $h_{i}^{LoS}$~\cite{Fath_2023}. The LoS channel gain from $i$-th luminaire to the centroid of the aperture of the lens can be expressed as
	\begin{equation}\label{equ:e10}
		h_{i}^{LoS}	=
		\displaystyle\frac{(m+1)A_{L}}{2\pi d_{i}^2}\cos^m(\theta_{i})\cos(\phi_{i})\Pi\left(\frac{\phi_i}{\phi_{FoV}}\right),
	\end{equation}
where $m=-\ln(2)/\ln(\cos(\theta_{1/2}))$ is the Lambertian order of the LED, $A_{L}$ is the aperture area of the lens, $d_{i}$ is the Euclidean distance between the centroid of the $i$-th LED and the centroid of the lens, and $\theta_{1/2}$ is the half-power semi-angle of the LEDs. Furthermore, $\theta_{i}$ is the irradiance angle of the $i$-th LED, $\phi_{i}$ is the incident angle at the lens from the $i$-th LED, $\phi_{FoV}$ is the field-of-view (FoV) of the PD, and $\Pi(x)$ is the rectangular function which is zero for $|x|>0.5$ while one for $|x|\le 0.5$. To obtain $h_{i,j}^{len}$, we use a geometric approach as explained in the Subsection \ref{subsec:h_len}.

\vspace{-3mm}
\subsection{Mobility and Random Receiver Orientation}\label{sec:mobility}
The practical conditions such as mobility and random receiver orientation have a significant impact on the performance of the liquid lens-assisted MIMO VLC systems. In order to account for such conditions and to model in a more practical sense, we use the following user mobility and random receiver orientation models.

\subsubsection{Mobility Model}
To model human motion in an indoor environment, clothoid and optimal control-based models have been used~\cite{Kapila_2023,Kapila_2024}. A clothoid satisfies the differential equations given by 
    \begin{equation}
    \dot{\mathbf{x}}(s)=\cos{\theta_R(s)}, 
    \end{equation}
    \begin{equation}
    \dot{\mathbf{y}}(s)=\sin{\theta_R(s)}, 
    \end{equation}
and
    \vspace{-2mm}
    \begin{equation}
    \dot{\mathbf{\theta}}_R(s)=\kappa_0 +\kappa_1 s,
    \end{equation}
with the following initial conditions, \textit{i.e.}, 
    \begin{equation}
    x(0) = x_0, \quad  y(0) = y_0, \quad \theta_R(0) = \theta_0,
    \end{equation}
where $(x_0,y_0)$ is the initial position of the user, and $\theta_0$ is the initial azimuth angle, respectively. The parameter $\kappa_0$ is the initial curvature, $\kappa_1$ is the sharpness of the curve and $s$ represents the curvilinear abscissa. From this system, the parametric expressions of a clothoid coordinate can be defined as
	\begin{equation}
			\mathbf{x}(s)=x_0+\int_0^s \cos\left(\theta_0+\mu+\frac{1}{2}\kappa_1\mu^2\right)d\mu, \\
	\end{equation}
        \vspace{-3mm}
and
    \vspace{-1mm}
	\begin{equation}
		\mathbf{y}(s)=y_0+\int_0^s \sin\left(\theta_0+\mu+\frac{1}{2}\kappa_1\mu^2\right)d\mu, \\
\end{equation}
where $\mu$ is the integration variable that is used to integrate from the initial point to the curve length $s$. To model the randomness of the user movements, we assume that user trajectories are generated through a random Gaussian process with zero mean and variance $\sigma_p^2$. 

\subsubsection{Random Receiver Orientation Model}
To model the random receiver orientation, we consider a realistic model as presented in~\cite{soltani_2019}. In addition, the polar angle of the mobile device is independent of the azimuth angle and the user position. According to the model, the polar angle of the receiver for the walking scenario follows a Gaussian distribution, and adjacent samples are time-correlated. To capture this time-correlation, a correlated Gaussian random process can be used. A first-order linear autoregressive (AR) model is used to generate the $k$-th sample of the polar angle, which is expressed as
	\begin{equation}
		\phi^{k}_R = c_0 + c_1\phi^{k-1}_{R}+w^{k},
  \label{eq:AR_model}
	\end{equation}
where $c_0$ is the biased level, $c_1$ is the constant factor of the AR, and $w^{k}$ is the white Gaussian noise with variance $\sigma_w^2$~\cite{soltani_2019}. The parameters of the AR model in~\eqref{eq:AR_model} can be calculated as 
\begin{equation}
	c_0 = (1-c_1)\mathbb{E}(\phi^{k}_R), 
 \end{equation}
 \begin{equation}
c_1 = \mathbb{R}_{\phi_R}\bigg(\frac{T_{c,\phi_R}}{T_s}\bigg)^{\frac{T_s}{T_{c,\phi_R}}},
\end{equation}
and 
\begin{equation} 
 \sigma_w^2=(1-c_1^2)\sigma_{\phi_R}^2, 
\end{equation}
where $\mathbb{E}(\phi^{k}_{R})$ denotes the mean value of $\phi^{k}_{R}$, $\mathbb{R}_{\phi_R}(\cdot)$ is the auto-correlation function, $T_{c,\phi_R}$ is the coherence time, $T_{s}$ is the sampling time, and $\sigma_{\phi_R}^2$ is the variance of $\phi^{k}_R$. The azimuth angle at the receiver and for the $k$-th instance (\textit{i.e.}, $\theta^{k}_R$) is the angle between the user's moving direction and the $x$-axis in the room's coordinate frame $OXYZ$.

\section{MIMO VLC Channel Characterization}\label{sec:channel}
In this section, we present an in-depth analysis of the channel model for the liquid lens-assisted MIMO VLC system by employing three-dimensional geometric modeling and geometric optics principles. Initially, we derive the rotation matrices, the unit normal vectors of the receiver and lens, and the spatial positions of the PDs relative to the room’s coordinate frame. Finally, the LoS light propagation, $h_{i}^{LoS}$, and the imaging channel gain, $h_{i,j}^{len}$, are evaluated.
\vspace{-3mm}
\subsection{Preliminary results}
In this section, we state some preliminary results, which will assist in the derivation of the main analytical framework. To begin with, the unit normal vector to the receiver, $\hat{\boldsymbol{\eta}}_{R}$, is derived in the following Lemma.
\begin{lemma}\label{Lemma1}
    The unit normal vector of the PD plane, $\hat{\boldsymbol{\eta}}_{R}$, is given by
    \begin{equation}\label{equ:eta_R}
    \hat{\boldsymbol{\eta}}_{R}
    = 
    \begin{bmatrix}
    \text{c}{\theta_R}\text{s}{\phi_R}\\
    \text{s}{\theta_R}\text{s}{\phi_R}\\ 
    \text{c}{\phi_R}
    \end{bmatrix},
\end{equation}
where $\theta_R$ and $\phi_R$ represent the rotation angles of the receiver's coordinate frame around its $Z$- and $Y$-axes, respectively, $\text{c}\theta = \cos{\theta}$, and $\text{s}\theta = \sin{\theta}$.
\end{lemma}
\begin{proof}
    See Appendix \ref{Appendix1}.
\end{proof}
In the following Remark, we calculate the coordinates of the $j$-th PD i.e., $\hat{\boldsymbol{P}}_j$ and the center of the lens i.e., $\hat{\boldsymbol{P}}_{len}$, w.r.t. the room's coordinates $OXYZ$.
\begin{remark}
    The coordinates of the $j$-th PD, $\hat{\boldsymbol{P}}_j$ and of the center of the lens, $\hat{\boldsymbol{P}}_{len}$, w.r.t. the frame $OXYZ$, are given by
    \begin{equation}
    \hat{\boldsymbol{P}}_j =
    \begin{bmatrix}
    x_{j}\\
    y_{j}\\
    z_{j}
    \end{bmatrix}
    = 
    \begin{bmatrix}
    x_R+x_{j}^R\text{c}{\theta_R}\text{c}{\phi_R}-y_j^{R}\text{s}{\theta_R}\\
    y_R+x_{j}^R\text{s}{\theta_R}\text{c}{\phi_R}+y_j^{R}\text{c}{\theta_R}\\ 
    z_R-x_{j}^R\text{s}{\phi_R}
    \end{bmatrix},
\end{equation}
and
\begin{equation}
    \hat{\boldsymbol{P}}_{len} =
    \begin{bmatrix}
    x_{len}\\
    y_{len}\\
    z_{len}
    \end{bmatrix}
    = 
    \begin{bmatrix}
    x_R+d_{len}\text{c}{\theta_R}\text{s}{\phi_R}\\
    y_R+d_{len}\text{s}{\theta_R}\text{s}{\phi_R}\\ 
    z_R+d_{len}\text{c}{\phi_R}
    \end{bmatrix},
\end{equation}
respectively, where $[x_j^R\quad y_j^R \quad 0]^T$ be the coordinates of the $j$-th PD of the imaging receiver w.r.t. receiver's coordinate frame $O'X'Y'Z'$.
\end{remark}
\begin{proof}
    Let $[x_j^R\quad y_j^R \quad 0]^T$ be the coordinates of the $j$-th PD of the imaging receiver w.r.t. receiver's coordinate frame $O'X'Y'Z'$. We assume that the $X'$ axis of the receiver lies in the direction of the user mobility {for mathematical tractability.} The position of the $j$-th PD w.r.t. the frame $OXYZ$ can be found by the relation $\hat{\boldsymbol{P}}_j = {^0\mathbf{R}_1^{-1}(\theta_R,\phi_R)[x_j^R\quad y_j^R \quad 0]^T+\hat{\boldsymbol{P}}_R}$. In addition, the coordinates of the center of the lens w.r.t. room's coordinate frame can be similarly calculated as $\hat{\boldsymbol{P}}_{len} = {^0\mathbf{R}_1^{-1}(\theta_R,\phi_R)[0\quad 0 \quad d_{len}]^T+\hat{\boldsymbol{P}}_R}$.
\end{proof}
The unit normal vector of the lens, $\hat{\boldsymbol{\eta}}_{len}$, is derived in the following Lemma.
\begin{lemma}\label{Lemma2}
    The unit normal vector of the lens, $\hat{\boldsymbol{\eta}}_{len}$, is given by
    \begin{equation}\label{equ:eta_len}
\hspace{-5mm}
    \hat{\boldsymbol{\eta}}_{len}
    = 
    \begin{bmatrix}
     \text{c}{\theta_R}\text{c}{\phi_R}\text{c}{\theta_L}\text{s}{\phi_L}-\text{s}{\theta_R}\text{s}{\theta_L}\text{s}{\phi_L}+\text{c}{\theta_R}\text{s}{\phi_R}\text{c}{\phi_L}\\
     \text{s}{\theta_R}\text{c}{\phi_R}\text{c}{\theta_L}\text{s}{\phi_L}+\text{c}{\theta_R}\text{s}{\theta_L}\text{s}{\phi_L}+\text{s}{\theta_R}\text{s}{\phi_R}\text{c}{\phi_L}\\
     -\text{s}{\phi_R}\text{c}{\theta_L}\text{s}{\phi_L}+\text{c}{\phi_R}\text{c}{\phi_L}\\
    \end{bmatrix},
\end{equation}
where $\theta_L$ and $\phi_L$ depict the rotation angles of the lens around its $Z'$- and $Y'$-axes by using its tilting mechanism.
\end{lemma}
\begin{proof}
    See Appendix \ref{Appendix2}.
\end{proof}

\begin{figure}[!t]
    \centering
    \includegraphics[width=0.55\columnwidth]{"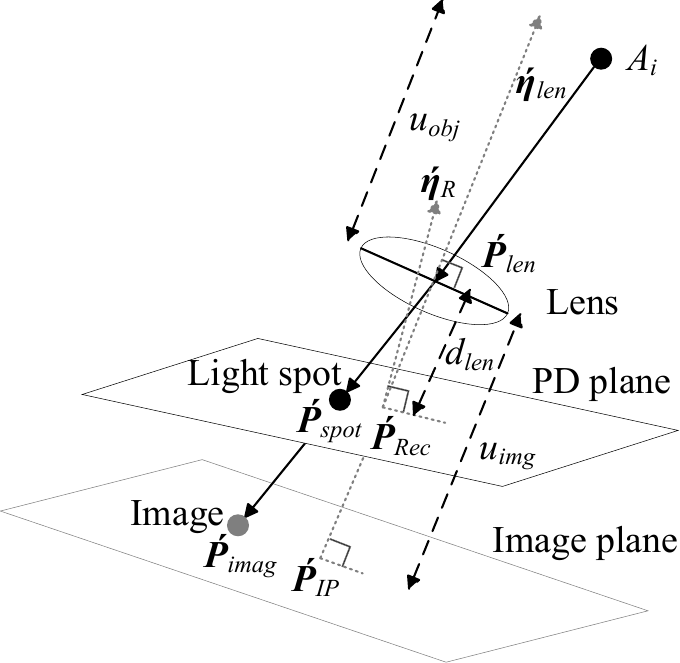"}
    \vspace{-2mm}
        \caption{Cross-sectional view of the light propagation from a point 
     on an LED and spot formation on the PD plane.}
    \label{fig:f3}
        \vspace{-5mm}
\end{figure}

\subsection{Calculation of $h_{i}^{LoS}$}
We firstly assess the LoS channel gain from the $i$-th luminaire to the centroid of the aperture of the lens. Towards this direction, $A_L$, $d_i$, $\theta_i$, and $\phi_i$ need to be calculated at each orientation and position instances of the receiver and the lens during the dynamic conditions. The effective area $A_L$ for a circular convex lens can be calculated with the simple relation $\pi r^2$, where $r = k_{\eta}f$ is the radius of the lens which is a function of $f$, and $k_{\eta}$ is a constant dependent on the refractive index of the liquid and the contact angle~\cite{Tian_2022}. The distance between the $i$-th luminary and the lens is $d_i= ||\hat{\boldsymbol{P}}_{i}-\hat{\boldsymbol{P}}_{len}||$. To calculate $\theta_i$ and $\phi_i$, first, we define the unit pointing vector from the midpoint of the lens to the midpoint of the $i$-th luminary that can be expressed as $\hat{\boldsymbol{\eta}}_{i,l}=\frac{\hat{\boldsymbol{P}}_{i}-\hat{\boldsymbol{P}}_{len}}{||\hat{\boldsymbol{P}}_{i}-\hat{\boldsymbol{P}}_{len}||}$. Now, the irradiance angle can be calculated as $\theta_i = \cos^{-1}\left(\hat{\boldsymbol{\eta}}_{i,l}\cdot[0 \quad 0 \quad 1]^T\right)$ and the incident angle is the angle between two unit vectors, $\hat{\boldsymbol{\eta}}_{i,l}$, and $\hat{\boldsymbol{\eta}}_{len}$ which is given by $\phi_i = \cos^{-1}\left(\hat{\boldsymbol{\eta}}_{i,l}\cdot\hat{\boldsymbol{\eta}}_{len}\right)$. Substituting these parameters into \eqref{equ:e10}, it can be expressed as
	\begin{equation}\label{equ:e101}
		h_{i}^{LoS}	=
		\displaystyle\frac{k_{LoS}f^2\left(z_i^L-z_R-d_{len}c\phi_R\right)^m}{ ||\hat{\boldsymbol{P}}_{i}-\hat{\boldsymbol{P}}_{len}||^{\frac{m+3}{2}}}\left(\hat{\boldsymbol{\eta}}_{i,l}\cdot\hat{\boldsymbol{\eta}}_{len}\right),
	\end{equation}
where $k_{LoS} = (m+1)k_{\eta}^2/2$ is a constant. 
    \vspace{-3mm}
\subsection{Calculation of $h_{i,j}^{len}$}\label{subsec:h_len}
In order to calculate $h_{i,j}^{len}$, the projected effective area of the $i$-th LED luminary on the $j$-th PD is required. Towards this direction, in the following Lemma, the coordinates of the focal point at the image plane and crossing point at the PD plane due to each vertex of $i$-th LED luminary are derived.
\begin{lemma}\label{Lemma3}
The coordinates of the focal point at the image plane and crossing point at the PD plane due to each vertex of $i$-th LED luminary are given by
\begin{equation}\label{equ:img2}
        \hat{\boldsymbol{P}}_{img,q_i} = \lambda_{ref}\hat{\boldsymbol{\eta}}_{ref,q_i}
 +\hat{\boldsymbol{\eta}}_{len},
    \end{equation}
        \vspace{-3mm}
and
    \begin{equation}\label{equ:img3}
        \hat{\boldsymbol{P}}_{spot,q_i} = \lambda_{spot}\hat{\boldsymbol{\eta}}_{ref,q_i}
 +\hat{\boldsymbol{\eta}}_{len},
    \end{equation}
respectively, where 
\begin{align}
        \hat{\boldsymbol{\eta}}_{ref,q_i} &= \frac{1}{n_{l}}\left[\hat{\boldsymbol{\eta}}_{len}\times (\hat{\boldsymbol{\eta}}_{len}\times \hat{\boldsymbol{\eta}}_{q_i,l}) \right] \nonumber \\
        &-\hat{\boldsymbol{\eta}}_{len}\sqrt{1-\frac{1}{n_{l}^2}(\hat{\boldsymbol{\eta}}_{len}\times\hat{\boldsymbol{\eta}}_{q_i,l})\cdot(\hat{\boldsymbol{\eta}}_{len}\times\hat{\boldsymbol{\eta}}_{q_i,l})},
    \end{align}
$n_1$ is the relative refractive index of the liquid, and $\lambda_{ref}$ is a constant.
\end{lemma}
\begin{proof}
See Appendix \ref{Appendix3}.
\end{proof}

This process can be iterated to obtain position vectors of $q_i, \forall i$ which are required to calculate $h_{i,j}^{len}$.

The area of the light spot formed for the $i$-th LED on the PD plane can be obtained by using the Shoelace formula\cite{Bart_1986}, and can be expressed as $\alpha_{i} = \frac{1}{2}\left(\lvert A'_i\times B'_i\rvert +\lvert B'_i\times C'_i\rvert+ \lvert C'_i\times D'_i\rvert+\lvert D'_i\times A'_i\rvert\right)$. Then, $h_{i,j}^{len}=\frac{\alpha_{i}\cap \beta_j}{\alpha_{i}}$, where $\beta_j$ is the area of the $j$-th PD, and $\alpha_{i}\cap \beta_j$ is the area of intersection between the light spot generated due the $i$-th LED and the $j$-th PD.

\section{BER Minimization}\label{sec:BER}
In this section, we formulate an optimization problem aimed at minimizing the BER by determining the optimal focal length and orientation angles of the proposed liquid lens. To address this problem, we propose two solution methodologies: a prediction-based blockwise ML framework and a nearest-LED selection strategy.

Initially, the exact BER corresponding to ML detection in \eqref{equ:e3} of the GSM-based VLC system is difficult to derive in closed-form. However, we derive a tight upper bound on the BER of the ML detector based on the pairwise error probability (PEP) analysis as in~\cite{chockalingam_2015}. The PEP which the receiver decides in favor of the signal vector $\mathbf{x}_m$ when $\mathbf{x}_n$ was transmitted can be written as~\cite{chockalingam_2015}
    \begin{equation}
        \mathrm{PEP}_{m,n} = \mathrm{PEP}(\mathbf{x}_{n}\rightarrow\mathbf{x}_{m}|\mathbf{H})=Q\left(\frac{\alpha r\Vert\mathbf{H}(\mathbf{x}_{m}-\mathbf{x}_{n})\Vert}{2\sigma}\right),
    \end{equation}
where $Q(x)=\frac{1}{\sqrt{2\pi}}\int_{x}^{\infty}\exp\left(-\frac{u^2}{2}\right)du $ is the Gaussian $Q$-function and the BER is upper bounded as~\cite{chockalingam_2015}
    \begin{align}\label{BER}
        \mathrm{BER} &\hspace{-0.5ex}\le \hspace{-0.5ex}\widetilde{\mathrm{BER}}\\
        &\hspace{-0.5ex}=\hspace{-0.5ex}\displaystyle\frac{1}{\eta_{gsm}2^{\eta_{gsm}}}\hspace{-1.5ex}\sum_{m=1}^{2^{\eta_{gsm}}}\hspace{-0.5ex}
        \sum_{n=1 \atop n\ne m}^{2^{\eta_{gsm}}} \hspace{-1ex}d_\mathbf{H}(\mathbf{x}_{m},\mathbf{x}_{n}) Q\left(\hspace{-0.5ex}\frac{r\Vert \mathbf{H}(\mathbf{x}_{m}\hspace{-0.5ex}-\hspace{-0.5ex}\mathbf{x}_{n})\Vert}{2\sigma}\hspace{-0.5ex}\right),\nonumber 
    \end{align}
where $d_\mathbf{H}(\mathbf{x},\mathbf{y})$ is the Hamming distance between $\mathbf{x}$ and $\mathbf{y}$. 

We consider the problem of minimizing the BER upper bound in \eqref{BER} with the focal length $f$, the azimuth angle $\theta_L$, and the polar angle $\phi_L$ of the reconfigurable liquid convex lens as the optimization variables. The corresponding optimization problem can be formulated as P1 \textit{i.e.,}
	\begin{mini!}|l|
		{\boldsymbol{p}, f,\theta_L, \phi_L}{\widetilde{\mathrm{BER}}(f,\theta_L,\phi_L)} 
		{\label{equ:e6}}{(\text{P1}) \quad \quad }
		\addConstraint{f^{min}\le f \le f^{max}}  
        \addConstraint{\theta_L^{min}\le \theta_L \le \theta_L^{max}}  
        \addConstraint{\phi_L^{min}\le \phi_L \le \phi_L^{max}},  
	\end{mini!}
where $\boldsymbol{p}$ denotes the indices of the certain LED activation.

\begin{figure*}[!t]
    \centering
    \includegraphics[width=0.8\textwidth]{"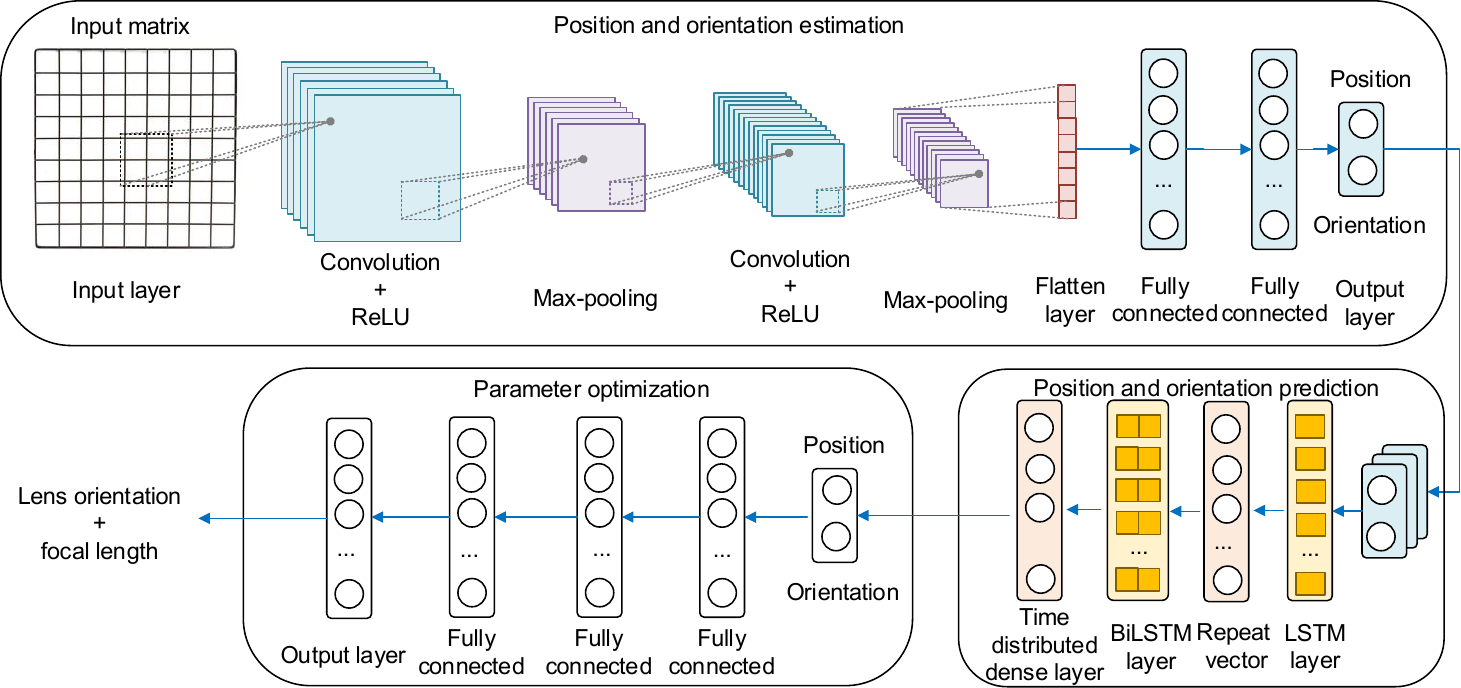"}
    \vspace{-2mm}
    \caption{Proposed CNN-LSTM-based ML architecture.}
    \label{fig:CNN}
    \vspace{-5mm}
\end{figure*}

Problem P1 is non-convex w.r.t. the optimization variables and can be solved using mixed monotonic programming~\cite{Matthiesen_2020}. Even though the number of optimization variables in P1 is small, this approach as well as other conventional optimization algorithms still have prohibitive computational complexity for online resource allocation. In particular, the channel gain calculations required to evaluate the objective of problem P1, are highly complex as they involve finding the intersection area of each light spot on the PD plane and each PD. Moreover, such calculations should be repeated for all possible $(f,\theta_L, \phi_L)$ pairs. On the other hand, prediction-based ML solutions are attractive for this problem, since they can predict optimized parameters for future time-instance, exploiting time correlated user positions and random receiver orientations-which is useful in online resource allocation~\cite{Kapila_2023}. The optimal parameter prediction for future time-instance is promising in this setup, as the system can adjust the lens with prior knowledge of optimal parameters, minimizing the communication performance gaps induced by the transient behavior of the lens. To this end, we propose two solution approaches including ($i$) the PBML scheme, and ($ii$) the CLS scheme, and subsequently compare their performance with two baseline schemes, namely vertical upward lens orientation (VULO) and exhaustive search.

\vspace{-1mm}
\subsection{PBML Scheme}

We present an artificial neural network (ANN)-based ML solution approach to obtain optimal lens parameters. To solve the P1 problem, different ML solutions including multi-layer perceptrons, convolution layers, long-short term memory (LSTM) layers, grated recurrent layers, and reinforcement learning can be used~\cite{Rekkas_2021}. However, our ML architecture is greatly motivated by presenting a lightweight ML architecture, that can limit the online resource allocation. Moreover, our ML approach is motivated to achieve twofold advances compared to other schemes, such as 1) the position and orientation can be estimated from the ML structure itself and no other position estimation mechanism is required\footnote{Since receiver position, receiver orientation, and channel states are not required for lens parameter optimization in PBML, it can be directly applied following the training process.}, and 2) predicting user positions and orientations for a future time instance and accordingly optimize system parameters for future time instance. To this end, a lightweight ANN comprised of three blocks is used as illustrated in Fig. \ref{fig:CNN}. The blockwise ANN is motivated to achieve higher training efficiency and to achieve better prediction accuracy~\cite{Zhong_2018}. The first block is responsible for position and orientation estimation. To this end, a convolutional neural network (CNN) which has several convolution layers followed by fully connected layers is used due to the CNN's ability to efficiently extract spatial features from a matrix~\cite{Zewen_2022}. The second block predicts the user position and orientation for a future time instance using previous samples of user position and orientation. This block is made with several LSTM layers, and bidirectional LSTM (BiLSTM) layers due to LSTM's/BiLSTM's ability to act as a memory and efficency in time series predictions~\cite{Sepp_1997}. The third block predicts the optimal lens orientation and the focal length for a given position and orientation using a ANN. The main input to the ANN is a matrix where each value represents the time averages received power at each PD. The outputs of the ANN include the orientation angles, the focal length of the liquid lens, and the LED activation pattern. In the following, we explain in detail each block of the ANN architecture.

\subsubsection{Position and Orientation Estimation}
This block estimates the user position and orientation by using the received average power matrix, denoted as $\mathbf{I}_{i}$ as the input. The value at the $(i_i,j_i)$-th index of the matrix $\mathbf{I}_{i}$ is calculated by averaging the received optical power at the $(i_i,j_i)$-th indexed PD over a time durations of $N_{TA}$ time slots. The dimension of the input matrix, $\mathbf{I}_{i}$, is $\sqrt{N_r}\times \sqrt{N_r}$. In order to extract spatial features from the received matrix, two convolution layers followed by two max-pooling layers, and flatten layer are used. First, $\mathbf{I}_{i}$ is 2D convoluted with $N_1$ number of $N_{k,1}\times N_{k,1}$ size kernels denoted by $\mathbf{K}_{c_1}$, where $c_1$ is the index of the kernal with no padding and stride equal to $1$, in order to generate $N_1$ number of features maps, each having the dimension of $(\sqrt{N_r}-N_{k,1}+1)\times (\sqrt{N_r}-N_{k,1}+1)$. From the 2D convolution, the $(i_{c_1},j_{c_1})$-th element of the $k_{c_1}$-th feature map can be expressed as
    \begin{align}\label{equ:conv1}
        \mathbf{F}_{k_{c_1}}[i_{c_1},j_{c_1}] &= \sum_{x_{i}=1}^{\sqrt{N_r}}\sum_{y_{i}=1}^{\sqrt{N_r}} \mathbf{K}_{c_1}[x_i,y_i]\nonumber \\
        &\times\mathbf{I}_{i}[i_{c_1}-x_i+1,j_{c_1}-y_i+1].
    \end{align}
Next, the rectified linear unit (ReLU) activation function is used, and the output is given by $\mathbf{F}^{ReLU}_{k_{c_1}}[i_{c_1},j_{c_1}] = \max\{\mathbf{F}_{k_{c_1}}[i_{c_1},j_{c_1}]+b_{i_{c_1},j_{c_1}},0\}$ where $b_{i_{c_1},j_{c_1}}$ is a bias. Output matrices of the first convolution layer, are then sent through an $N_1$ number of max-pooling layers of the dimension $N_{m,1}\times N_{m,1}$ with stride $N_{m,1}$ and the output dimensions are $\frac{1}{N_{m,1}}(\sqrt{N_r}-N_{k,1}+1)\times \frac{1}{N_{m,1}}(\sqrt{N_r}-N_{k,1}+1)$. In the second convolution layer, the outputs are convoluted with $N_2$ number of $N_{k,2} \times N_{k,2}$ size kernels with no padding and stride equal to 1, to generate $N_1\times N_2 $ number of convolution layers of size $(\frac{1}{N_{m,1}}(\sqrt{N_r}-N_{k,1}+1)-N_{k,2}+1)\times (\frac{1}{N_{m,1}}(\sqrt{N_r}-N_{k,1}+1)-N_{k,2}+1)$. From the 2D convolution, the $(i_{c_2},j_{c_2})$-th element of the $k_{c_2}$-th feature map, $\mathbf{F}_{k_{c_2}}[i_{c_2},j_{c_2}]$ can be expressed similar to \eqref{equ:conv1}, by replacing $i_{c_1}$ with $i_{c_2}$, $j_{c_1}$ with $j_{c_2}$, $\sqrt{N_r}$ with $\frac{1}{N_{m,1}}(\sqrt{N_r}-N_{k,1}+1)$, $\mathbf{K}_{c_1}$ with $\mathbf{K}_{c_2}$, and $\mathbf{I}_{i}$ with $\mathbf{I}_{c_2}$, where $\mathbf{K}_{c_2}$ is the kernel matrix, and $\mathbf{I}_{c_2}$ is the input matrices to the second convolution layer, respectively. Next, the ReLU activation function is used, and the output is $\mathbf{F}^{ReLU}_{k_{c_2}}[i_{c_2},j_{c_2}] = \max\{\mathbf{F}_{k_{c_2}}[i_{c_2},j_{c_2}]+b_{i_{c_2},j_{c_2}},0\}$. The generated matrices are sent through $N_1\times N_2$ second set of max-pooling layers which has the dimension of $N_{m,2}\times N_{m,2}$ with stride $N_{m,2}$, and hence, the output dimension is $\frac{1}{N_{m,2}}(\frac{1}{N_{m,1}}(\sqrt{N_r}-N_{k,1}+1)-N_{k,2}+1)\times \frac{1}{N_{m,2}}(\frac{1}{N_{m,1}}(\sqrt{N_r}-N_{k,1}+1)-N_{k,2}+1)$. Then, the generated feature map is converted to an 1D array of the size $N_1N_2(\frac{1}{N_{m,2}}(\frac{1}{N_{m,1}}(\sqrt{N_r}-N_{k,1}+1)-N_{k,2}+1))^2$ using a flatten layer. Next, the feature map is sent through two fully connected layers with a ReLU activation function and of sizes $N_{f,1}$ and $N_{f,2}$ to obtain the current receiver position and orientation. The number of nodes in the output layer is five. That includes $(x,y,z)$ coordinates of the receiver position as well as the $\theta_R$ and $\phi_R$.

\subsubsection{Position and Orientation Prediction}
This block is designed to predict the user position and orientation for a future time instance. In particular, we collect $N_I$ samples of estimated user position and orientation values upto the $n$-th time instance obtained from the position and orientation block as inputs and predict the user position and orientation for the $(n+1)$-th time instance. The input layer consists of $5\times N_I$ nodes while the output layer consists $5$ nodes. In order to obtain the user position and the orientation for the $(n+1)$-th time instance, a recurrent neural network consisting of a long-short term memory (LSTM) layer with $N_{l,1}$ elements, a repeat vector with $N_{r,1}$ nodes, a bidirectional LSTM (BiLSTM) with $N_{l,2}$ elements, and a time distributed dense layer with $N_{d,1}$ is used. The use of BiLSTM layers is helpful to obtain predicted user locations efficiently. In BiLSTM layers, two LSTMs are used for forward and backward propagation. Finally, the output layer provides $(x,y,z)$ coordinates of the receiver position as well as the $\theta_R$ and $\phi_R$ for $n$-th time instance.

\subsubsection{Parameter Optimization}
In this block, the predicted user positions and receiver orientations obtained from the position and orientation prediction block are used as inputs to obtain an optimized lens orientation and LED activation pattern. For this purpose, a regression model is implemented, where three dense layers with $N_{D1}$, $N_{D2}$, and $N_{D3}$ nodes are used in a cascade. The ReLU activation function is used in each layer. The output layer consists of $3$ nodes to obtain optimal $\theta_L$, $\phi_L$, and $f$.

The outputs of all dense layers are denoted as $\hat{\mathbf{x}}^{(D1)}\in\mathbb{R}^{N_{D1}}$, $\hat{\mathbf{x}}^{(D2)}\in\mathbb{R}^{N_{D2}}$, and $\hat{\mathbf{x}}^{(D3)}\in\mathbb{R}^{N_{D3}}$. They can be written as 

\begin{equation}
\hat{\mathbf{x}}^{(D1)}=f_{ReLU}\left(\mathbf{W}_{D1}\hat{\mathbf{x}}^{(5)}+\mathbf{k}_{D1}\right),
\end{equation}
\vspace{-5mm}
\begin{equation}
\hat{\mathbf{x}}^{(D2)}=f_{ReLU}\left(\mathbf{W}_{D2}\hat{\mathbf{x}}^{(D1)}+\mathbf{k}_{D2}\right),
\end{equation}
and
\begin{equation}
\hat{\mathbf{y}}=f_{ReLU}\left(\mathbf{W}_{D3}\hat{\mathbf{x}}^{(D2)}+\mathbf{k}_{D3}\right),
\end{equation}
where $\hat{\mathbf{x}}^{(5)}$ is the vector that includes predicted user positions and receiver orientations, $\mathbf{W}_{D1}\in\mathbb{R}^{(5\times N_{D1})}$, $\mathbf{W}_{D2}\in\mathbb{R}^{(N_{D1}\times N_{D2})}$, and $\mathbf{W}_{D3}\in\mathbb{R}^{(N_{D2}\times N_{D3})}$ are the weight vectors, $\mathbf{k}_{D1}\in\mathbb{R}^{N_{D1}}$, $\mathbf{k}_{D2}\in\mathbb{R}^{N_{D2}}$, and $\mathbf{k}_{D3}\in\mathbb{R}^{N_{D3}}$ are the bias vectors, and $\hat{\mathbf{y}}$ is the output vector.

\subsubsection{Training Phase}
In the training phase, block-wise training is considered. To train the position and orientation estimation block, a set of average received power matrices sampled during the user mobility and the time-correlated random receiver orientation as described in the Section \ref{sec:mobility} are used and corresponding $(x,y,z)$ coordinates of the receiver, $\theta_R$, and $\phi_R$ are used as training labels. To train the position and orientation prediction block, user positions and receiver orientation in a window of size $N_I$ are considered as inputs, while user positions and receiver orientations at the next time slot are considered as training labels. The regression model for parameter optimization is trained by considering user positions and receiver orientations at each time slot as inputs, while taking optimal $\theta_L$, $\phi_L$, and $f$ obtained from an exhaustive search for training labels. In the exhaustive search, to find the optimal parameters, an iterative search technique that includes coarse and fine-tuning was used.

To achieve training in each block, initial values for the network weights $\boldsymbol{\theta}$ are set. Forward propagation is applied to obtain the output vector $\hat{\mathbf{y}}$ for a selected input from the training set. The mean squared error (MSE) between $\hat{\mathbf{y}}$ and a vector containing optimum parameters $\mathbf{y_0}$ is used as the loss function, which is expressed as 
\begin{equation}
\phi(\boldsymbol{\theta}) = \frac{1}{B}\sum_{i_m=1}^{B}\Vert \hat{\mathbf{y}}- \mathbf{y_0}\Vert^2 , 
\end{equation}%
where $B$ is the mini-batch size of the training, $\mathbf{y_0}$ includes optimum training labels obtained from the exhaustive search, and $\boldsymbol{\theta}$ is updated for training batches using the stochastic gradient descent algorithm that can be expressed as $\boldsymbol{\theta}^+ := \boldsymbol{\theta}-\epsilon_L\nabla\phi(\boldsymbol{\theta})$, where $\epsilon_L$ is the learning rate.

It is noted that the position and orientation estimation block, and the parameter optimization can be trained one time prior to the use and no retraining is required again as the trained ML architecture is valid for a given LED placement and receiver architecture. However, the position and orientation prediction block may need frequent updates or retraining depending on the indoor environment, since this block's performance depends on human mobility and receiver orientation patterns. Upon the complete restructuring of the indoor environment (\textit{e.g.,} rearrangement of the lab environment, class orientation etc.) updates may be required to achieve high accuracy. However, our blockwise architecture avoids retraining the whole model.

\subsubsection{Computational Complexity Analysis}
Now, we explain the computational complexity involved with the PBML scheme. To this end, the number of calculations involved in each block in terms of multiplication operations and summation operations are presented. The total number of calculations involved in the block $B_i$ is expressed as $N_{B_i} = N_{B_i}^{mul} + N_{first}^{sum}$ in which the $N_{B_i}^{mul}$ is the number of multiplication operations, $N_{B_i}^{sum}$ is the summation operations, and $B_i= 1, 2,\text{and }3$ denotes the position and orientation estimation, the positions and orientation prediction, and parameter optimization blocks, respectively. $N_{1}^{mul}$ (excluding activation functions, max-pooling operations, and flattening operations which are low complex) is expressed as~\cite{Freire_2024}
\begin{align}\label{cc1}
    &N_{1}^{mul} = N_1N_{k,1}^2(\sqrt{N_r}-N_{k,1}+1)^2\nonumber \\
    &+N_1N_2N_{k,2}^2\left(\frac{1}{N_{m,1}}(\sqrt{N_r}-N_{k,1}+1)-N_{k,2}+1\right)^2\nonumber\\
    & + N_{f,1}N_1N_2\left(\frac{1}{N_{m,1}N_{m,2}}(\sqrt{N_r}-N_{k,1}+1)-N_{k,2}+1\right)^2 \nonumber \\
    &+ N_{f,1}N_{f,2}+5N_{f,2},
\end{align}
and $N_{1}^{sum}$ is
\begin{align}\label{cc2}
    N_{1}^{sum}\hspace{-0.5mm} &= \hspace{-0.5mm}N_1(\sqrt{N_r}-N_{k,1}+1)^2+\hspace{-0.5mm}N_{f,1}\hspace{-0.5mm}+\hspace{-0.5mm}N_{f,2}\hspace{-0.5mm}+\hspace{-0.5mm}5\\
    &+\hspace{-0.5mm}N_1N_2\left(\hspace{-0.5mm}\frac{1}{N_{m,1}}(\sqrt{N_r}-N_{k,1}+1)-N_{k,2}+1\hspace{-0.5mm}\right)^2\hspace{-1.5mm}.\nonumber
\end{align}
$N_{2}^{mul}$ (excluding activation functions) is expressed as~\cite{Freire_2024}
\begin{align}\label{cc3}
    N_{2}^{mul} \hspace{-0.5mm}&= \hspace{-0.5mm}4(6N_{l,1}+N_{l,1}^2)\hspace{-0.5mm}+\hspace{-0.5mm}8N_{r,1}(N_{l,1}(N_{l,1}\hspace{-0.5mm}+\hspace{-0.5mm}N_{l,2})+N_{l,2}^2)\nonumber\\
    &+5N_{l,2},
\end{align}
and $N_{2}^{sum}$ is
\begin{align}\label{cc4}
    N_{2}^{sum} \hspace{-0.5mm}=\hspace{-0.5mm} 4(6N_{l,1}\hspace{-0.5mm}\hspace{-0.5mm}+N_{l,1}^2)\hspace{-0.5mm}+\hspace{-0.5mm} 8N_{r,1}(N_{l,1}(N_{l,1}\hspace{-0.5mm}+\hspace{-0.5mm}N_{l,2})\hspace{-0.5mm}+\hspace{-0.5mm}N_{l,2}^2)\hspace{-0.5mm}+\hspace{-0.5mm}5.
\end{align}
$N_{3}^{mul}$ (excluding activation functions) can be expressed as\cite{Freire_2024}
\begin{align}\label{cc5}
    N_{3}^{mul} = 5N_{D1}+N_{D1}N_{D2} + N_{D2}N_{D3}+ 3N_{D3},
\end{align}   
and $N_{3}^{sum}$ is expressed as
\begin{align}\label{cc6}
    N_{3}^{sum} = N_{D1}+ N_{D2} + N_{D3}+3.
\end{align}
We note that despite expressions in \eqref{cc1}, \eqref{cc2}, \eqref{cc3}, \eqref{cc4}, \eqref{cc5}, and \eqref{cc6} are complex, the number of nodes in each layer are finite and small, and hence, the proposed PBML scheme can perform accurately in online optimization in a short time duration.

\subsection{CLS Scheme}
Towards achieving analytical tractability, this subsection introduces a low-complexity solution aimed at simplifying the performance evaluation process. This is motivated by the fact that the objective function of problem P1 is equivalent to maximizing the term, $\Vert \mathbf{H}(\mathbf{x}_{m}-\mathbf{x}_{n})\Vert$ where $\mathbf{x}_m$ and $\mathbf{x}_n$ are two distinct transmit signal vectors from the set of all possible transmit signal vectors. Hence, the sparsity of $\mathbf{H}(\alpha,\theta_L,\phi_L)$ needs to be increased. Furthermore, we observe that the objective function is often decreased the most by increasing the largest value of $\mathbf{H}(\alpha,\theta_L,\phi_L)$ even further. The largest value of the channel gain matrix $\mathbf{H}$ corresponds to the shortest LED-receiver link, and when the lens is oriented towards the closest LED, it maximizes this gain, thereby optimizing the system's BER performance. Based on this observation, we adjust the orientation angles of the lens, $\theta_L$ and $\phi_L$, such that the axis of the lens is aligned with the closest LED to the lens. This is achieved by selecting the LED that satisfies the condition
\begin{align}\label{l1}
	i^*=\argmin_{i}\ d_{i},
	\end{align}
where $d_i$ is the Euclidean distance from the $i$-th LED to the mid-point of the lens of the receiver. This alignment ensures that the lens focuses on the LED that provides the strongest received signal, improving the overall system performance. In the following Lemma, analytical expressions for the orientation angles are derived, conditioned on the location of the closest LED.
\begin{lemma}\label{Lemma4}
in the context of the proposed CLS scheme, the rotation angle of the lens around the $Y'$-and $Z'$-axes are given by
\begin{align}\label{eqs:quad1_sol_sim}
\phi_L = \cos^{-1}{\left(\text{c}{\phi_R}\left(\frac{z_{i^*}-z_{len}}{d_{i^*}}\right)\right)},
\end{align}
and
\vspace{-4mm}
\begin{align}\label{eqs:quad1_sol_sim2}
    \displaystyle \theta_L = \sin^{-1}\left(\frac{\text{c}{\theta_R} \left(\frac{y_{i^*}-y_{len}}{d_{i^*}}\right)- \text{s}{\theta_R}\left(\frac{x_{i^*}-x_{len}}{d_{i^*}}\right)}{\sqrt{1-\text{c}{\phi_R}\left(\frac{z_{i^*}-z_{len}}{d_{i^*}}\right)^2}}\right),
\end{align}
respectively, where $(x_{i^*}, y_{i^*},z_{i^*})$ depict the coordinate of the selected $i^*$-th LED, and $d_{i^*}= \sqrt{(x_{i^*}-x_{len})^2+(y_{i^*}-y_{len})^2+(z_{i^*}-z_{len})^2}$ denotes the distance to the closest LED from the lens.
\end{lemma}
\begin{proof}
    See Appendix \ref{Appendix4}.
\end{proof}

In the context of the proposed CLS scheme, the focal length, $f$, is selected such that the light coming from the closest LED that goes through the lens is focused on to the plane of the PD array. In other words, $f$ is adjusted such that the distance between the lens and the light spot formed due to the closest LED is equal to $f$. This selection is motivated by the fact that the channel gain from the closest LED to a PD reaches its maximum value when its light spot is focused optimally on the PD plane, thereby increasing the norm in the BER expression and improving the BER value. Hence, the selected focal length in this scheme is
\vspace{-2mm}
\begin{align}\label{eqs:focal_len1}
f^* = \frac{d_{len}}{\hat{\boldsymbol{\eta}}_{len}\cdot \hat{\boldsymbol{\eta}}_{R}}.
\end{align}
With the help of \eqref{equ:eta_R} and \eqref{equ:eta_len}, \eqref{eqs:focal_len1} can be simplified as
\begin{align}\label{eqs:focal_len2}
f^* \hspace{-0.5ex}= \hspace{-0.5ex}\frac{d_{len}}{\text{c}{\theta_R}\text{s}{\phi_R}\left(\hspace{-0.5ex}\frac{x_{i^*}-x_{len}}{d_{i^*}}\hspace{-0.5ex}\right)\hspace{-0.5ex}+\hspace{-0.5ex} \text{s}{\theta_R}\text{s}{\phi_R}\left(\hspace{-0.5ex}\frac{y_{i^*}-y_{len}}{d_{i^*}}\hspace{-0.5ex}\right)\hspace{-0.5ex}+\hspace{-0.5ex} \text{c}{\phi_R}\left(\hspace{-0.5ex}\frac{z_{i^*}-z_{len}}{d_{i^*}}\hspace{-0.5ex}\right)}.
\end{align}

\vspace{-6mm}
\subsection{Benchmark Schemes}
In this subsection, we present two benchmark schemes to rigorously evaluate the effectiveness of the proposed optimization frameworks. The first scheme, denoted as VULO, maintains the lens in a fixed vertically upward position, regardless of user location or receiver orientation, thus serving as an extremely low-complexity solution. Conversely, the second benchmark, referred to as Exhaustive Search, provides the optimal solution by evaluating all possible configurations, thereby achieving the best performance at the expense of substantially increased computational complexity.

\subsubsection{VULO Scheme}
We consider orienting the axis of the lens of the receiver vertically upward along the $z$ axis of the room coordinate frame, that can be expressed as $\hat{\boldsymbol{\eta}}_{len} = [0 \quad 0 \quad 1]^T$. By obtaining three expressions for $x$, $y$, $z$ direction vectors, the following equations hold \textit{i.e.},
\begin{align}\label{eqs:362}
\text{c} {\theta_R}\text{c}{\phi_R}\text{c}{\theta_L}\text{s}{\phi_L}-\text{s}{\theta_R}\text{s}{\theta_L}\text{s}{\phi_L}+
\text{c}{\theta_R}\text{s}{\phi_R}\text{c}{\phi_L} = 0. 
\end{align}
\vspace{-3mm}
\begin{align}\label{eqs:372}
\text{s}{\theta_R}\text{c}{\phi_R}\text{c}{\theta_L}\text{s}{\phi_L}+\text{c}{\theta_R}\text{s}{\theta_L}\text{s}{\phi_L}+
\text{s}{\theta_R}\text{s}{\phi_R}\text{c}{\phi_L} = 0.
\end{align}
\vspace{-3mm}
\begin{align}\label{eqs:382}
-\text{s}{\phi_R}\text{c}{\theta_L}\text{s}{\phi_L}+\text{c}{\phi_R}\text{c}{\phi_L} = 1.
\end{align}
Next, we multiply \eqref{eqs:362} by $\text{c}\theta_R$ and \eqref{eqs:372} by $\text{s}\theta_R$. Then, by taking the addition of two resulting expressions, the following expression  is obtained \textit{i.e.},
\begin{align}\label{eqs:391}
\text{c}{\phi_R} \text{c}{\theta_L} \text{s}{\phi_L}+ \text{s}{\phi_R} \text{c}{\theta_L} = 0.
\end{align}
By substituting \eqref{eqs:382} into \eqref{eqs:391} and after some mathematical manipulations, we get two solutions such as \textit{Case 1:} $\phi_L = \phi_R$, and \textit{Case 2:} $\phi_L = -\phi_R$. Next, by substituting $\phi_R$ to \eqref{eqs:382}, the corresponding $\theta_L$ values can be obtained as \textit{Case 1:} $\theta_L = 180^{\circ}$, and \textit{Case 2:} $\theta_L = 0$.

In this scheme, $f^*$ is set equal to the distance from the lens to the receiver plane along the $z$ axis of the room's coordinate frame. This selection is motivated, since the vertical downward light coming to the lens is focused well on to the receiver plane in this scheme. The selected focus length in this scheme can be expressed as 
	\begin{align}\label{fl2}
	f^* = \frac{d_{len}}{\cos{\phi_R}}.
	\end{align}
 
\subsubsection{Exhaustive Search}
In this approach, the solution to problem P1 is found using an exhaustive search algorithm. The objective of P1 is calculated for possible combinations of the solution space and the optimal solution is found using a iterative approach as explained below. First, we divide the range of focal length, $f^{min}\le f \le f^{max}$, into $N_{f}$ discrete points, the range of azimuth angle of the lens, $\theta_{L}^{min}\le\theta_{L}\le\theta_{L}^{max}$, into $N_{\theta_L}$ discrete points, and the range of polar angle of the lens, $\phi_{L}^{min}\le\phi_{L}\le\phi_{L}^{max}$, into $N_{\phi_L}$ discrete points. The BER values for each sample point is calculated by using \eqref{BER}. Next, the minimum BER is found from the sample set. To fine tune the optimal $f$, $\theta_L$, and $\phi_L$ values, we iterate this step with smaller resolution around the optimal parameters obtained in the previous step until the BER improvement is less than a threshold, $\epsilon_{BER}$. 

	\begin{table}[!t]
		\caption{Simulation parameters.}
		\label{Tab:para}
		\centering\begin{tabular}{|c|c|c|c|c|c|}
			\hline
			Symbol & Value & Symbol & Value & Symbol & Value\\
			\hline
			$\rho$ & $1.5$ & $\sigma$ & $10^{-6}$ & $N_{k,2}$ & 2\\
			\hline
                $A_{L}$ & $1 \times 10^{-4}$ m$^{2}$ & $I_P$ & $1$ & $N_{m,2}$ & 1\\
                \hline
                $\Phi_{FoV}$ & $90^\circ$ & $f^{min}$ & $1$ cm & $N_{1}$ & 20 \\
                \hline
                $\theta_{1/2}$ & $60^\circ$ & $f^{max}$ & $15$ cm & $N_{2}$ & 20\\
                \hline
                $r$ & $0.75$~A/W & $\theta_L^{min}$ & $0^\circ$ & $N_{f,1}$ & 20\\
                \hline
                $N_t = N_r$ & $16$ & $\theta_L^{max}$ & $360^\circ$ & $N_{f,2}$ & 20\\
                \hline
                $\alpha$ & $1$ W/A & $\phi_{min}$ & $0^\circ$ & $N_{I}$ & 10\\
                \hline
                $d_{len}$ & $0.02$~m & $\phi_{max}$ & $30^\circ$ & $N_{l,1}$ & 10 \\
                \hline
                $d_s$ & $0.25$~m & $\sigma_{\phi_R}^2$ & $10^\circ$ & $N_{l,2}$ & 10 \\
                \hline
                $d_{tx}$ & $0.5$~m& $T_{c,\phi_R}$ & $1$ ms & $N_{d,1}$ & 10\\
                \hline
                $d_{rx}$ & $5$ mm & $k_{\eta}$ & $0.1$ & $N_{D1}$ & 30\\
                \hline
                $M$ & $2$ & $N_{k,1}$ & 2 & $N_{D2}$ & 40\\
                \hline
                $N_a$ & $2$ & $N_{m,1}$ & 1 & $N_{D3}$ & 30\\
                \hline
		\end{tabular}
      \vspace{-3mm}
	\end{table}

 \vspace{-2mm}
\section{Numerical Results and Discussion}\label{sec:results}
 
We present numerical results to verify the performance gains of the liquid lens-assisted imaging receiver-based MIMO VLC system and presented optimization schemes. The results are compared with MIMO VLC systems with/without static convex lenses and imaging receivers. A room dimension of $5\text{m}\times 5\text{m}\times3.5\text{m}$ is assumed. Unless stated explicitly, otherwise in all simulations, we have set the parameter values as given in Table \ref{Tab:para}. For the training of the ML model, a set of $10^7$ channel realizations were generated. A series of simulations were conducted to obtain suitable percentages for training, validation, and testing and were selected as $70\%$, $10\%$, and $20\%$. The evaluation was performed by MATLAB and Python running on a DELL XPS 15 7590 with a $4\times2.4$ GHz Intel Core i7 CPU and an NVIDIA GeForce GTX 1650 4GB GPU. The average execution time for the ML architecture was found to be 9.4 ms.

\begin{figure}[!t]
    \centering
    \includegraphics[width=0.99\columnwidth]{"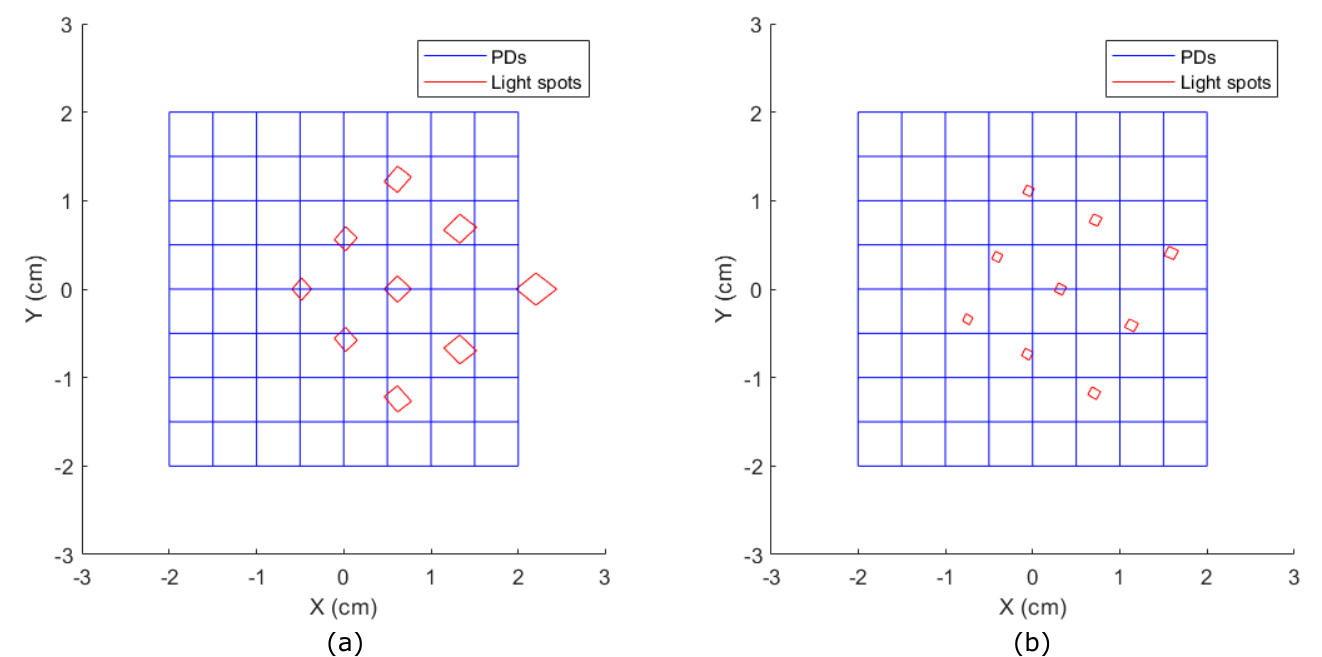"}
    \vspace{-3mm}
    \caption{Light spots obtained on the PD array using liquid convex lens at the receiver position $(2.5, 2.5, 0)$, and the receiver orientation, $\theta_R = 45^\circ$, and $\phi_R = 17^\circ$. (a) Without optimization (fixed lens parameters): $\theta_L = 5^\circ$, $\phi_L = {25}^\circ$, and $f = 3$ cm. (b) With optimization (using PBML scheme): $\theta_L = 18.2^\circ$, $\phi_L = 8.1^\circ$, and $f = 2.12$ cm.}
    \label{fig:result1}
        \vspace{-5mm}
\end{figure}

\textit{Effect of Lens Parameters: }Fig. \ref{fig:result1} shows the effect of the liquid lens parameter optimization on the light spot locations on the PD plane. The receiver is placed on the location $(2.5, 2.5, 0)$ in the room, and the receiver orientation angles are $\theta_R = 45^\circ$, and $\phi_R = 17^\circ$, respectively. Fig. \ref{fig:result1}(a) corresponds to the case where lens angles are fixed at $\theta_L = {5}^\circ$, and $\phi_L = {25}^\circ$ while the focal length is $f = 3$ cm. It can be noticed that one light spot is partially outside the area of the PD array, for the considered instance. Also, some PDs can be partially covered by more than one light spot, which results in high co-channel interference and weak BER performance of $8.2\times 10^{-2}$. On the other hand, Fig. \ref{fig:result1}(b) shows the case where lens parameters are optimized using the PBML scheme and are $\theta_L = 18.2^\circ$, $\phi_L = 8.1^\circ$, and $f = 2.12$ cm, respectively. In this configuration, all the light spots are captured by PDs, focused on sharp light beams, and low co-channel interference is observed. Such an optimization of lens parameters results in an improved BER of $0.72\times 10^{-3}$. Even though we show this performance gains in a certain instance of user mobility and random receiver orientation, a significant BER improvement can be observed in most of the positions and receiver orientation conditions. The performance gap can vary depending on the user position and the receiver orientation.

\textit{Position and Orientation Prediction Accuracy: }In Fig. \ref{fig:path}(a), we show an example of the user mobility patterns obtained using clothoid functions and predicted user positions obtained using the position and orientation prediction block of the ML architecture. Results are helpful to verify that our CNN-based position and orientation estimation and position prediction approach by using BiLSTM and LSTM are helpful to obtain the user positions and receiver orientation for future time instance with reasonable accuracy. In addition, \ref{fig:path}(b) shows the average MSE of the position and orientation prediction versus window size of the input used for position and orientation prediction, $N_I$. The average MSE of the position and orientation prediction is obtained as
\vspace{-2mm}
    \begin{align} 
    &MSE=\frac{1}{N_TN_P}\sum_{i_p=1}^{N_P}\sum_{k=1}^{N_T}\bigg\{(x_{i_p}^{k}-\hat x_{i_p}^{k})^2+(y_{i_p}^{i_c}-\hat y_{i_p}^{i_c})^2+\nonumber\\
    &(z_{i_p}^{k}-\hat z_{i_p}^{k})^2
    +(\theta_{R,i_p}^{k}-\hat \theta_{R,i_p}^{k})^2+(\phi_{R,i_p}^{k}-\hat \phi_{R,i_p}^{k})^2\bigg\},
        \vspace{-1mm}
    \end{align} 
where $(x_{i_p}^{k},y_{i_p}^{k},z_{i_p}^{k})$ and $(\hat x_{i_p}^{k},\hat y_{i_p}^{k}, \hat z_{i_p}^{k})$ are the actual and the predicted $k$-th position sample of the $i_p$-th user mobility path, respectively, $(\theta_{R,i_p}^{k}, \phi_{R, i_p}^{k})$ and $(\hat \theta_{R,i_p}^{k},\hat \phi_{R,i_p}^{k})$ are the actual and the predicted $k$-th receiver orientation pair sample of the $i_p$-th user mobility path, respectively, $N_T$ is the number of samples in each user mobility path, and $N_P$ is the number of paths considered. We illustrate the results for different numbers of PDs, $N_r$. Results show that the average MSE decreases as the $N_I$ increases, due to higher prediction accuracy with a larger window size. However, results show a saturating trend at high $N_I$, in particular when $N_I>10$, the improvement of MSE is small. On the other hand, a large $N_I$ value increases the computational complexity of the ML block. Hence, $N_I = 10$ can be selected as a practical value for implementation to obtain the required prediction accuracy. However, this value depends on the available hardware and accuracy requirements of the system. Moreover, an increase of $N_r$ results in better prediction as CNN can efficiently extract features from a larger input matrix, resulting in more accurate position and receiver orientation values\footnote{The works in~\cite{Dehghani_2023} and~\cite{Arin_2019} have presented imaging receivers with large PD arrays with small receiver area that can be easily adapted to our system.}. Specifically, an MSE reduction from $2\times 10^{-2}$ to $8\times 10^{-3}$ can be observed when $N_r$ increases from $16$ to $25$ at $N_I = 10$. However, the gap between curves reduces as $N_r$ increases. 

\begin{figure}[!t]
    \centering
    \includegraphics[width=0.99\columnwidth]{"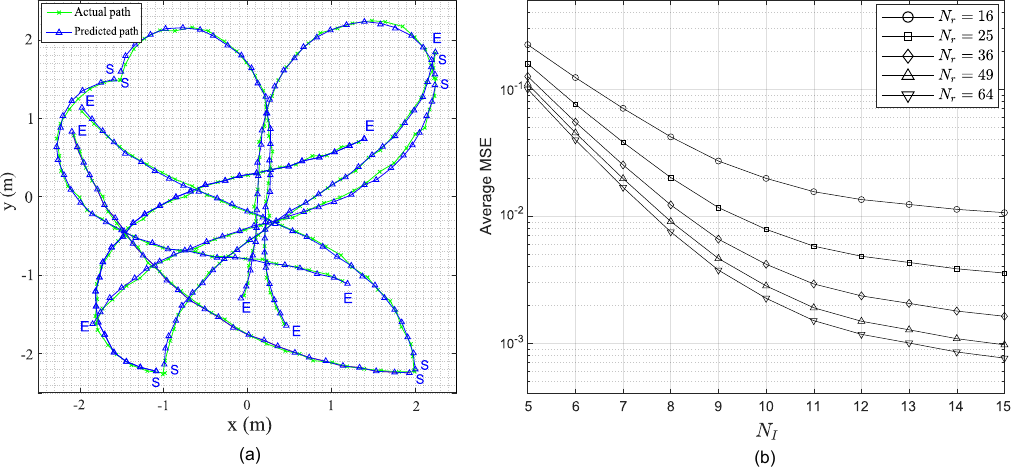"}
        \vspace{-3mm}
    \caption{(a) An example mobility pattern and path prediction in the indoor environment. S and E denote the starting and ending points, respectively. (b). Average MSE of the path prediction vs. $N_I$.}
    \label{fig:path}
        \vspace{-3mm}
\end{figure}

\begin{figure}[!t]
    \centering
    \includegraphics[width=0.92\columnwidth]{"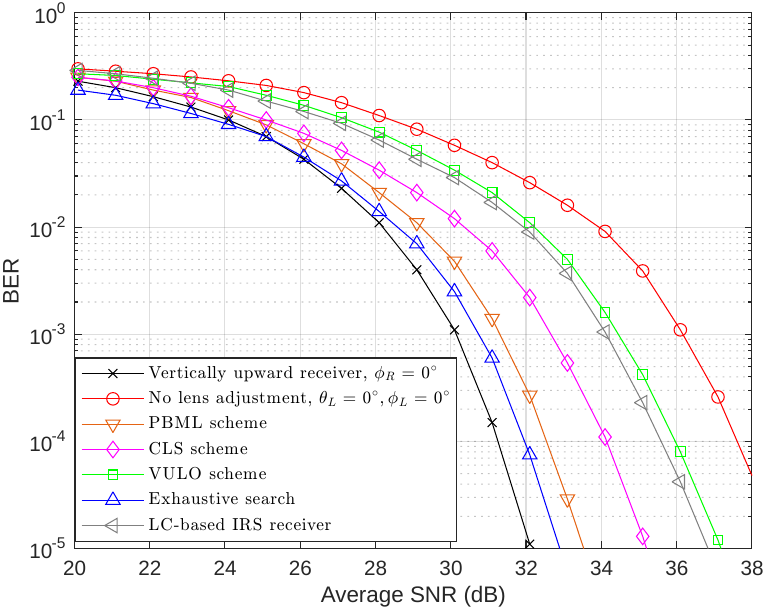"}
        \vspace{-2mm}
    \caption{\textcolor{black}{BER vs. average SNR for different optimization schemes.}}
    \label{fig:result2}
        \vspace{-5mm}
\end{figure}

\textit{Performance Comparison of Optimization Schemes: }
Fig. \ref{fig:result2} shows the BER performance of proposed two optimization schemes versus average SNR in dB obtained using \eqref{equ:e2}. The results are compared with the two benchmark schemes. In addition, to show the advances of using liquid lenses and the proposed optimization schemes, the results with no lens adjustment, denoted by $\theta_L = 0^{\circ}$, $\phi_L = 0^{\circ}$, and $f = 3$ cm, and with vertical upward receiver orientation, denoted by $\phi_R = 0^{\circ}$ are also presented.  Results show that when the receiver is vertically upward with no receiver orientation, the BER performance is much better compared to the case when there is random receiver orientation. The random receiver orientation results in poor BER performance when no lens adjustment is applied. Our two lens optimization schemes help to significantly improve the BER performance. Specifically, the BER can be improved from $6\times 10^{-2}$ to $1.4\times 10^{-3}$ at an average SNR of $30$ dB.This represents a significant gain under conditions of user mobility and random receiver orientation. Further improvements can be achieved with receivers that exhibit extremely low noise generation, such as those based on single-photon avalanche PDs. Among all schemes, exhaustive search achieves a better BER which is close to the case without any receiver orientation. The proposed PBML scheme shows near optimal results, with only a $0.9$ dB gap in BER at $30$ dB SNR value. In addition to that, CLS scheme is helpful to improve the BER as low complexity optimization schemes where the BER improvement is limited. Results indicate that VULO scheme has the lowest impact on the BER performance, however, the BER gap increases at a high SNR. Moreover, we compare our schemes with the liquid crystal (LC)-based IRS receiver presented in~\cite{Ngatched_2021}. To this end, we have selected the LC-based IRS receiver parameters as the weighted percentage of terthiophene (3T-2 MB) as $8\%$, the weighted percentage of trinitrofluorenone as $0.01\%$, the temperature as $30^{\circ}\text{C}$, LC cavity depth as $10 \mu\text{m}$, and external electric field as $1.2$ (V/$\mu$m)~\cite{Ngatched_2021}. Results illustrated in Fig. \ref{fig:result2} show that the BER performance is better than no lens adjustment. However, our proposed liquid convex lens-based receiver performs better than LC-based receivers, except with the VULO scheme.

\begin{figure}[!t]
    \centering
    \hspace{2.5mm}\includegraphics[width=0.89\columnwidth]{"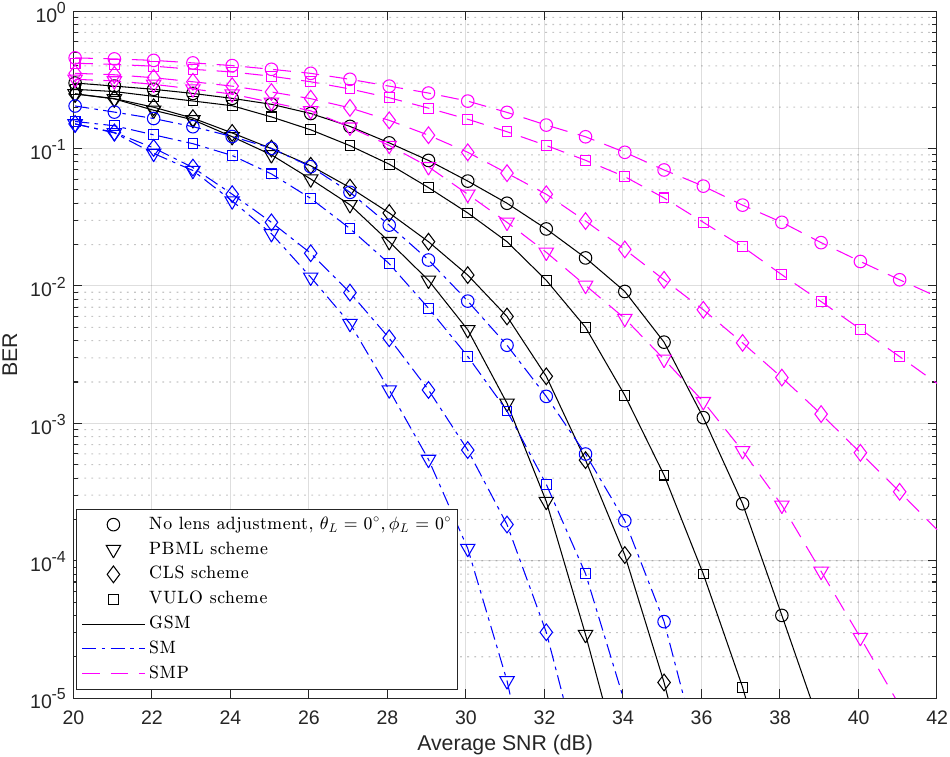"}
        \vspace{-2mm}
    \caption{\textcolor{black}{BER vs. average SNR for different MIMO schemes.}}
    \label{fig:result2_2}
        \vspace{-5mm}
\end{figure}

\textit{\textcolor{black}{Effect of different MIMO schemes:} }
\textcolor{black}{Fig. \ref{fig:result2_2} illustrates the BER performance versus SNR for various lens optimization schemes, PBML, CLS, VULO, and no lens optimization applied to different MIMO transmission techniques, namely GSM, SM, and SMP. The results demonstrate that PBML consistently achieves the best BER performance across all MIMO schemes. Among the transmission schemes, SM exhibits the lowest BER, albeit at the expense of reduced spectral efficiency. In contrast, SMP shows the lowest BER performance. Notably, the relative performance gain of the PBML scheme over the baseline (i.e., without lens optimization) is most pronounced in the SMP configuration. This improvement is attributed to PBML’s ability to reduce channel correlation and enhance the effective rank of the channel matrix, which is especially beneficial in the highly interference-limited SMP scenario.}

\begin{figure}[!t]
    \centering
    \includegraphics[width=0.92\columnwidth]{"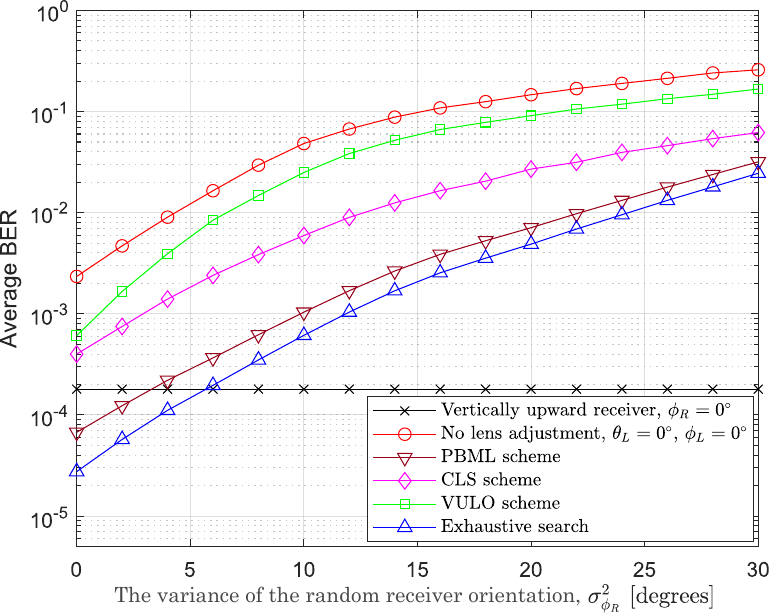"}
    \caption{BER vs. the variance of the random receiver orientation, $\sigma_{\phi_R}^2$, for different optimization schemes.}
    \label{fig:result3}
    \vspace{-5mm}
\end{figure}

\textit{Effect of the Random Receiver Orientation: }
To show the superiority of the reconfigurable liquid lens-assisted MIMO VLC system in different random receiver conditions, Fig. \ref{fig:result3} illustrates the BER performance versus the variance of the random receiver orientation, $\sigma_{\phi_R}^2$, for proposed optimization schemes, and benchmark schemes. A high $\sigma_{\phi_R}^2$ value results in a weak BER performance in all optimization schemes as the channel gain values become smaller, and hence, the norm in the BER expression becomes small. With no lens adjustment, the BER performance is weak compared to all other schemes and saturates fast around $\sigma_{\phi_R}^2 = 15^{\circ}$ as light spots are not focused properly on each PD. The results are helpful to identify that the use of liquid lens and lens parameters optimization results in better performance in a wide range of $\sigma_{\phi_R}^2$ and the use of PBML scheme is helpful to obtain near optimal performance. In particular, a BER improvement of more than $9.5$ dB can be observed for $\sigma_{\phi_R}^2 \le 30^{\circ}$. In addition, for $\sigma_{\phi_R}^2 \le 5^{\circ}$, even a better performance than a static receiver is observed with liquid lens and optimization using PBML scheme or exhaustive search. As low complexity schemes, CLS scheme and VULO are also helpful to improve BER in a wide range of $\sigma_{\phi_R}^2$. In particular, exhaustive search and CLS scheme have similar performance gains for low $\sigma_{\phi_R}^2$ values. This is expected due to the selection of almost similar lens parameter values from exhaustive search and CLS scheme when $\phi_R$ is closer to zero.

\begin{figure}[!t]
    \centering
    \includegraphics[width=0.92\columnwidth]{"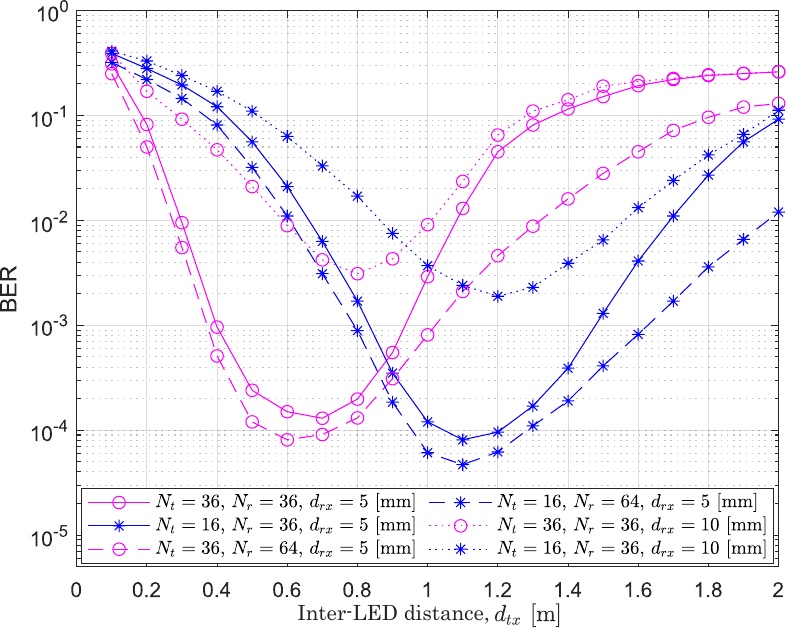"}
    \caption{BER vs. inter-LED distance, $d_{tx}$, obtained using PBML scheme. The average SNR is $31$ dB, $\sigma^2_{\phi_R} = 10^{\circ}$, and $N_a = 1$.}
    \label{fig:result4}
    \vspace{-5mm}
\end{figure}

\textit{Effect of the Inter-LED Distance: }
In Fig. \ref{fig:result4}, we show the BER performance versus inter-LED distance, $d_{tx}$. For this simulation, we used $\sigma^2_{\phi_R} = 10^{\circ}$, $N_a = 1$, and the average SNR was $31$ dB. Results are illustrated for different number of LED transmitters, $N_t$, different number of PD receivers, $N_r$, and different inter-PD distance, $d_{rx}$ values. Our results indicate that at low $d_{tx}$, the BER performance is poor due to the overlap of light spot locations on the PD plane, which leads to high correlation among the channels. Similarly, high $d_{tx}$ also results in poor BER, particularly at high $N_t$, as average channel gains are reduced, and channel gains from edge LEDs become exceedingly poor. Therefore, for a given $N_t$, there exists an optimal $d_{tx}$ value that minimizes the BER. Specifically, the optimal $d_{tx}$ value are $0.7$ m and $1.1$ m for the cases with $N_t = 36$ and $N_t = 16$, respectively. The optimal $d_{tx}$ tends to decrease as $N_t$ increases. This behavior is attributed not only to the spatial arrangement of LEDs and PDs, but also due to the adaptive liquid lens's ability to focus light from a large area of the ceiling onto the PD plane. Consequently, our results provide valuable insights for selecting practical values in real-world implementation\footnote{Optimizing the placement of LEDs presents an intriguing direction for future research.}. It is noted that a higher $N_r$ value results in better BER performance, especially at high $d_{tx}$. This is expected due to two factors such as ($i$) more PDs can capture light coming from LEDs at the corner of the room, and ($ii$) the position and orientation prediction accuracy of the ML architecture increases as more features can be extracted using the CNN. In addition, high $d_{rx}$ results in poor performance as a sparsely placed PD array cannot capture formed light spots from the liquid lens efficiently. This result is helpful to identify that a closely packed large number of PDs in the PD plane can be helpful in achieving maximum performance of the liquid lens-assisted GSM-based MIMO VLC systems. 

\vspace{-3mm}
\section{Conclusion}\label{sec:conclusion}
In this paper, we proposed an adjustable liquid lens-assisted imaging receiver for the MIMO VLC system. To achieve efficient communication, the GSM scheme is adopted in this setup. Specifically, practical conditions such as user mobility and random receiver orientation were taken into account. A channel gain model applicable for liquid lens-assisted MIMO VLC is presented that can be applied in future system designs. We optimized the focal length and orientation angles of the liquid lens in order to minimize the BER performance. Driven by the fact that channel model of the liquid lens-assisted MIMO VLC system is mathematically complex, two optimization schemes including ($i$) PBML scheme (ML-based), and ($ii$) CLS scheme (low-complexity scheme) were presented and compared with two benchmark schemes and conventional static imaging receivers. In particular, the PBML scheme can estimate the user position and receiver orientation, predict them for future time instance, and accordingly optimize lens parameters. Our results show that the proposed liquid lens-based imaging receiver and the proposed ML-based optimization scheme are helpful in achieving superior BER performance. Specifically, the BER can be improved from $6\times 10^{-2}$ to $1.4\times 10^{-3}$ at an average SNR of $30$ dB using the PBML scheme. Further, our results facilitate to identify the optimal system parameters for the proposed model and performance gains of the optimization schemes under a wide range of receiver orientation conditions.

\vspace{-3mm}
\appendices
\section{Proof of Lemma \ref{Lemma1}}\label{Appendix1}
The receiver's coordinate frame, $O'X'Y'Z'$ can be considered as a rotated version of the room's coordinate frame, $OXYZ$ by $\theta_R$ around its $Z$-axis and $\phi_R$ by its $Y$-axis. The rotation matrix between the frames $OXYZ$ and $O'X'Y'Z'$ can be expressed as  $^0\mathbf{R}_1(\theta_R,\phi_R)=\mathbf{R}_Y(\phi_R)\mathbf{R}_Z(\theta_R)$, where $\mathbf{R}_Y(\phi_R)$ is the rotation matrix around $Y$-axis by $\phi_R$ angle and $\mathbf{R}_Z(\theta_R)$ is the rotation matrix around $Z$-axis by $\theta_R$ angle and it reduces to
\begin{equation}\label{equ:rot01}
    ^0\mathbf{R}_1(\theta_R,\phi_R)
    = 
    \begin{bmatrix}
     \text{c}{\theta_R}\text{c}{\phi_R} & \text{s}{\theta_R}\text{c}{\phi_R} & -\text{s}{\phi_R}\\
     -\text{s}{\theta_R} & \text{c}{\theta_R}& 0\\
    \text{c}{\theta_R}\text{s}(\phi_R) & \text{s}{\theta_R}\text{s}{\phi_R} & \text{c}{\phi_R}\\
    \end{bmatrix},
\end{equation}
where $\text{c}\theta = \cos{\theta}$, and $\text{s}\theta = \sin{\theta}$. The unit normal vector to the PD plane lies along the axis, $Z'$, and hence, by using the passive rotation between two coordinate frames $OXYZ$ and $O'X'Y'Z'$, the final expression can be derived by evaluating the expression $\hat{\boldsymbol{\eta}}_{R}={^0\mathbf{R}_1^{-1}(\theta_R,\phi_R)[0\quad 0 \quad 1]^T}$.

\vspace{-3mm}
\section{Proof of Lemma \ref{Lemma2}}\label{Appendix2}
The lens has two degrees of freedom, having the ability to rotate around its $Z'$-axis to change its yaw angle $\theta_L$ by rotating its magnetic ring. It can be tilted around its $Y'$-axis to create a polar angle $\phi_L$ using its tilting mechanism. Taking these into account, the rotation matrix from the receiver coordinate frame to the lens coordinate frame can be expressed as $^1\mathbf{R}_2(\theta_R,\phi_R)=\mathbf{R}_{Y'}(\phi_L)\mathbf{R}_{Z'}(\theta_L)$, where $\mathbf{R}_{Y'}(\phi_L)$ is the rotation matrix around $Y'$-axis by $\phi_L$ angle and $\mathbf{R}_{Z'}(\theta_L)$ is the rotation matrix around $Z'$-axis by $\theta_L$ angle. The unit normal vector of the lens can be thought of as a unit vector along its own $Z''$-axis, and it can be represented in the room's coordinate frame as $\hat{\boldsymbol{\eta}}_{len}=^0\mathbf{R}_1^{-1}(\theta_R,\phi_R)^1\mathbf{R}_2^{-1}(\theta_L,\phi_L)[0\quad 0 \quad 1]^T$.

\vspace{-3mm}
\section{Proof of Lemma \ref{Lemma3}}\label{Appendix3}
Let $A_i$, $B_i$, $C_i$, and $D_i$ represent the vertices of the $i$-th LED transmitter,
with corresponding image points on the PD plane as $A'_i$, $B'_i$, $C'_i$, and $D'_i$ (See Fig. \ref{fig:f1}(b)). The points $A_i$, $B_i$, $C_i$, and $D_i$ can be expressed as $\boldsymbol{P}_{A_i} = (x_i^L-d_{ts}/2, y_i^L-d_{ts}/2, z_i^L)$, $\boldsymbol{P}_{B_i} = (x_i^L-d_{ts}/2, y_i^L+d_{ts}/2, z_i^L)$, $\boldsymbol{P}_{C_i} = (x_i^L+d_{ts}/2, y_i^L+d_{ts}/2, z_i^L)$, and $\boldsymbol{P}_{D_i} = (x_i^L+d_{ts}/2, y_i^L-d_{ts}/2, z_i^L)$. The unit vector along a direction from the midpoint of the lens to the point $q_i\in\{A_i,B_i,C_i,D_i\}$ is given by
$\hat{\boldsymbol{\eta}}_{q_i,l}=\frac{\hat{\boldsymbol{P}}_{q_i}-\hat{\boldsymbol{P}}_{len}}{||\hat{\boldsymbol{P}}_{q_i}-\hat{\boldsymbol{P}}_{len}||}$.
Similar to the case of $\phi_i$, the incident angle on the lens from the $q_i$-th point can be expressed as $\phi_{q_i} = \cos^{-1}\left(\hat{\boldsymbol{\eta}}_{q_i,l}\cdot\hat{\boldsymbol{\eta}}_{len}\right)$. Moreover, the projection of the line connecting the mid-point of the lens to the point $q_i$ along the axis $\hat{\boldsymbol{\eta}}_{len}$ is $d_{q_i, \hat{\boldsymbol{\eta}}_{len}} = ||(\hat{\boldsymbol{P}}_{q_i}-\hat{\boldsymbol{P}}_{len})|| \cos\left(\phi_{q_i}\right)$.

Assuming the use of the paraxial optical lens, the image formed by the lens can be geometrically distortionless in comparison to the source. Let's assume the paraxial distance to the object from the lens is $u_{obj} = d_{q_i, \hat{\boldsymbol{\eta}}_{len}}$, and the paraxial distance to the image plane is $u_{img}$. By using the basic lens maker's law given by $\frac{1}{f} = \frac{1}{u_{img}}-\frac{1}{u_{obj}}$, it can be rearranged as $\frac{u_{obj}}{u_{img}} = \frac{u_{obj}}{f}+1$. This expression can be used to calculate the the magnification produced by the lens which can be expressed as $m = \frac{u_{img}}{u_{obj}} = \frac{f}{f+d_{q_i, \hat{\boldsymbol{\eta}}_{len}}}$. In our setup, it is not possible to assume that the spot formation is on the image plane since the focal length of the lens can be adjusted. However, since the object's distance is much larger than the focal length, it can be assumed that the object is at infinity, and hence the image plane is closer to the focal plane. With these assumptions, the cross point of the axis of the lens with the image plane is expressed as
    \begin{equation}\label{equ:pip}
        \hat{\boldsymbol{P}}_{IP} = \hat{\boldsymbol{P}}_{len} - m \left(\hat{\boldsymbol{P}}_{T,q_i}-\hat{\boldsymbol{P}}_{len}\right),
    \end{equation}
where $\hat{\boldsymbol{P}}_{T,q_i} = \hat{\boldsymbol{P}}_{len}+d_{q_i, \hat{\boldsymbol{\eta}}_{len}}\hat{\boldsymbol{\eta}}_{len}$ is the position vector of the point on the axis of the lens project by the point $q_i$. The image plane is perpendicular to the axis of the lens, and hence, the position of the focal point of the light on the image plane is expressed as 
    \begin{equation}\label{equ:img}
        \left(\hat{\boldsymbol{P}}_{img,q_i} - \hat{\boldsymbol{P}}_{IP}\right)\cdot\hat{\boldsymbol{\eta}}_{len} = 0.
    \end{equation}
Considering the small area of the lens and approximating to a single reflection from the lens surface, the direction of the refracted light beam can be expressed as~\cite{Amit_2012}
    \begin{align}
        \hat{\boldsymbol{\eta}}_{ref,q_i} &= \frac{1}{n_{l}}\left[\hat{\boldsymbol{\eta}}_{len}\times (\hat{\boldsymbol{\eta}}_{len}\times \hat{\boldsymbol{\eta}}_{q_i,l}) \right] \nonumber \\
        &-\hat{\boldsymbol{\eta}}_{len}\sqrt{1-\frac{1}{n_{l}^2}(\hat{\boldsymbol{\eta}}_{len}\times\hat{\boldsymbol{\eta}}_{q_i,l})\cdot(\hat{\boldsymbol{\eta}}_{len}\times\hat{\boldsymbol{\eta}}_{q_i,l})},
    \end{align}
where $n_1$ is the relative refractive index of the liquid. We assume that the PD plane is close to the image plane and a single crossing point of the light with the PD plane is considered. Hence, $\hat{\boldsymbol{P}}_{img,q_i}$ can be further expressed as    
    \begin{equation}\label{equ:img2}
        \hat{\boldsymbol{P}}_{img,q_i} = \lambda_{ref}\hat{\boldsymbol{\eta}}_{ref,q_i}
 +\hat{\boldsymbol{\eta}}_{len},
    \end{equation}
where $\lambda_{ref}$ is a constant. Substituting \eqref{equ:img2} and \eqref{equ:pip} into \eqref{equ:img}, $\lambda_{ref}$ can be obtained, and hence, $\hat{\boldsymbol{P}}_{img,q_i}$ can be resolved. The locations of the light spots on the PD plane can be obtained as follows. 
The PD plane is perpendicular to the axis of the receiver, and hence, the position of the light spot on the PD plane is expressed as 
    \begin{equation}\label{equ:pdplane}
        \left(\hat{\boldsymbol{P}}_{spot,q_i} - \hat{\boldsymbol{P}}_{R}\right)\cdot\hat{\boldsymbol{\eta}}_{Rec} = 0.
    \end{equation}
Also, $\hat{\boldsymbol{P}}_{spot,q_i}$ lies on the reflected ray and is expressed as
    \begin{equation}\label{equ:img3}
        \hat{\boldsymbol{P}}_{spot,q_i} = \lambda_{spot}\hat{\boldsymbol{\eta}}_{ref,q_i}
 +\hat{\boldsymbol{\eta}}_{len},
    \end{equation}
where $\lambda_{spot}$ is a constant. Substituting \eqref{equ:img3} into \eqref{equ:img}, $\lambda_{spot}$ can be obtained, and hence, $\hat{\boldsymbol{P}}_{spot,q_i}$ can be resolved.

\section{Proof of Lemma \ref{Lemma4}}\label{Appendix4}
Let $d_{i^*}= \sqrt{(x_{i^*}-x_{len})^2+(y_{i^*}-y_{len})^2+(z_{i^*}-z_{len})^2}$ denotes the distance to the closest LED from the lens, where $(x_{i^*}, y_{i^*},z_{i^*})$ as the coordinate of the selected $i^*$-th LED. Equating coefficients along $x$, $y$, and $z$ directions in $\hat{\boldsymbol{\eta}}_{i^*,l}$ to $\hat{\boldsymbol{\eta}}_{len}$ the following equations hold,
\begin{align}\label{eqs:36}
\text{c} {\theta_R}\text{c}{\phi_R}\text{c}{\theta_L}\text{s}{\phi_L}-\text{s}{\theta_R}\text{s}{\theta_L}\text{s}{\phi_L}+
\text{c}{\theta_R}\text{s}{\phi_R}\text{c}{\phi_L} = \frac{x_{i^*}-x_{len}}{d_{i^*}}, 
\end{align}
    \vspace{-5mm}
\begin{align}\label{eqs:37}
\text{s}{\theta_R}\text{c}{\phi_R}\text{c}{\theta_L}\text{s}{\phi_L}+\text{c}{\theta_R}\text{s}{\theta_L}\text{s}{\phi_L}+
\text{s}{\theta_R}\text{s}{\phi_R}\text{c}{\phi_L} = \frac{y_{i^*}-y_{len}}{d_{i^*}},
\end{align}
    \vspace{-5mm}
\begin{align}\label{eqs:38}
-\text{s}{\phi_R}\text{c}{\theta_L}\text{s}{\phi_L}+\text{c}{\phi_R}\text{c}{\phi_L} = \frac{z_{i^*}-z_{len}}{d_{i^*}}.
\end{align}
The expressions in \eqref{eqs:36}, \eqref{eqs:37}, and \eqref{eqs:38} can be combined to obtain closed-form expressions for $\theta_L$ and $\phi_L$ for this scheme. For this purpose, \eqref{eqs:36} is multiplied with $\text{s} \theta_R$ and \eqref{eqs:37} is multiplied with $\text{c} \theta_R$. The resulting expressions can be subtracted to get a single expression. The resulting expression and \eqref{eqs:38} are combined, and after some mathematical manipulations, a quadratic equation of the variable $\text{c}{\phi_L}$ can be obtained and expressed as
\begin{align}\label{eqs:quad1}
&\text{c}{\phi_L}^2-2\text{c}{\phi_R}\left(\frac{z_{i^*}-z_{len}}{d_{i^*}}\right)\text{c}{\phi_L} + \left(\frac{z_{i^*}-z_{len}}{d_{i^*}}\right)^2 
+ \text{s}{\phi_R}^2\times\nonumber \\
&\left(\text{c}{\theta_R} \hspace{-0.5ex} \left(\frac{y_{i^*}-y_{len}}{d_{i^*}}\right)\hspace{-0.5ex}-\hspace{-0.5ex} \text{s}{\theta_R}\left(\frac{x_{i^*}-x_{len}}{d_{i^*}}\right)\right)^2
\hspace{-0.5ex}- \hspace{-0.5ex}\text{s}{\phi_R}^2\hspace{-0.5ex}= \hspace{-0.5ex}0. 
\end{align}
After finding the solutions for the quadratic equation in \eqref{eqs:quad1}, $\text{c}{\phi_L}$ is expressed as 
\begin{align}\label{eqs:quad1_sol}
&\text{c}{\phi_L} = \text{c}{\theta_R}\left(\frac{z_{i^*}-z_{len}}{d_{i^*}}\right)\pm\text{s}{\theta_R}\Bigg(1-\left(\frac{z_{i^*}-z_{len}}{d_{i^*}}\right)^2\nonumber \\
&-\left(\text{c}{\theta_R} \left(\frac{y_{i^*}-y_{len}}{d_{i^*}}\right)- \text{s}{\theta_R}\left(\frac{x_{i^*}-x_{len}}{d_{i^*}}\right)\right)^2\Bigg)^{\frac{1}{2}}.
\end{align}
However, from the geometrical observations it is noted that the square-root component of \eqref{eqs:quad1_sol} is zero and $\phi_L$ has a distinct solution, which is the final expression for $\phi_L$ in Lemma \ref{Lemma3}. Finally, by using \eqref{eqs:36}, \eqref{eqs:37}, and \eqref{eqs:quad1_sol_sim}, with the trigonometric relation $\sin{\theta} = \sqrt{1-\cos^2{\theta}}$, the final expression for $\theta_L$ can be derived.

\let\mybibitem\bibitem
\renewcommand{\bibitem}[1]{%
  \ifstrequal{#1}{}
    {\color{blue}\mybibitem{#1}}
    {\ifstrequal{#1}{}{\color{blue}\mybibitem{#1}}{\ifstrequal{#1}{}{\color{blue}\mybibitem{#1}}{\color{black}\mybibitem{#1}}}}
}

\vspace{-1mm}
\bibliographystyle{IEEEtran}
\bibliography{IEEEabrv,Main}

\begin{IEEEbiography}[{\includegraphics[width=1in,keepaspectratio]{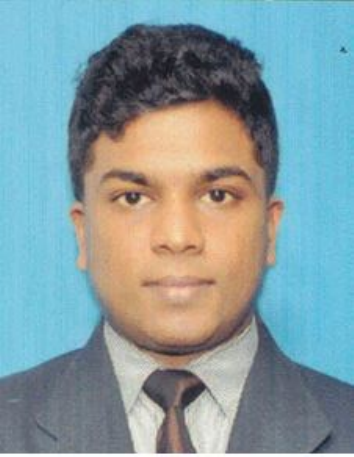}}]{Kapila W. S. Palitharathna}
(Member, IEEE) received the B.Sc. Eng. (Hons.) and Ph.D. degrees in Electrical and Electronic Engineering from the University of Peradeniya, Sri Lanka, in 2016, and in 2022. He is currently a Postdoctoral Research Associate at the IRIDA Research Centre for Communication Technologies, Department of Electrical and Computer Engineering, University of Cyprus. He has worked as a Research Assistant, a Senior Lecturer, and the Head of the Department (Telecommunication) at the Sri Lanka Technological Campus, Sri Lanka, from 2018 to 2023. He also worked as an Instructor at the University of Peradeniya, Sri Lanka, from 2016 to 2017. His research interests include visible light communication, underwater optical wireless communication, non-orthogonal multiple access, intelligent reflective surfaces, energy harvesting communications, and machine learning for communication. 
\end{IEEEbiography}

\begin{IEEEbiography}[{\includegraphics[width=1in,keepaspectratio]{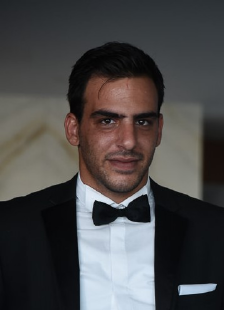}}]{Christodoulos Skouroumounis} (S’15–M’20) received the Diploma in Computer Engineering from the Electrical and Computer Engineering Department of the National Technical University of Athens (NTUA), Greece, in 2014, and a Ph.D. in Computer Engineering from the University of Cyprus, Cyprus, in 2019. He is currently a Senior Wireless R\&D Engineer at CyRIC – Cyprus Research and Innovation Center Ltd. His main research interests include the modeling, design, and performance analysis of future wireless communication systems, using tools such as multi-antenna network information theory and stochastic geometry. Specific topics include cooperative networks, full-duplex radio, millimeter-wave/THz communications, wireless powered communications, joint communication and radar sensing networks, satellite communications, and reconfigurable fluid antennas.
\end{IEEEbiography}

\begin{IEEEbiography}[{\includegraphics[width=1in,keepaspectratio]{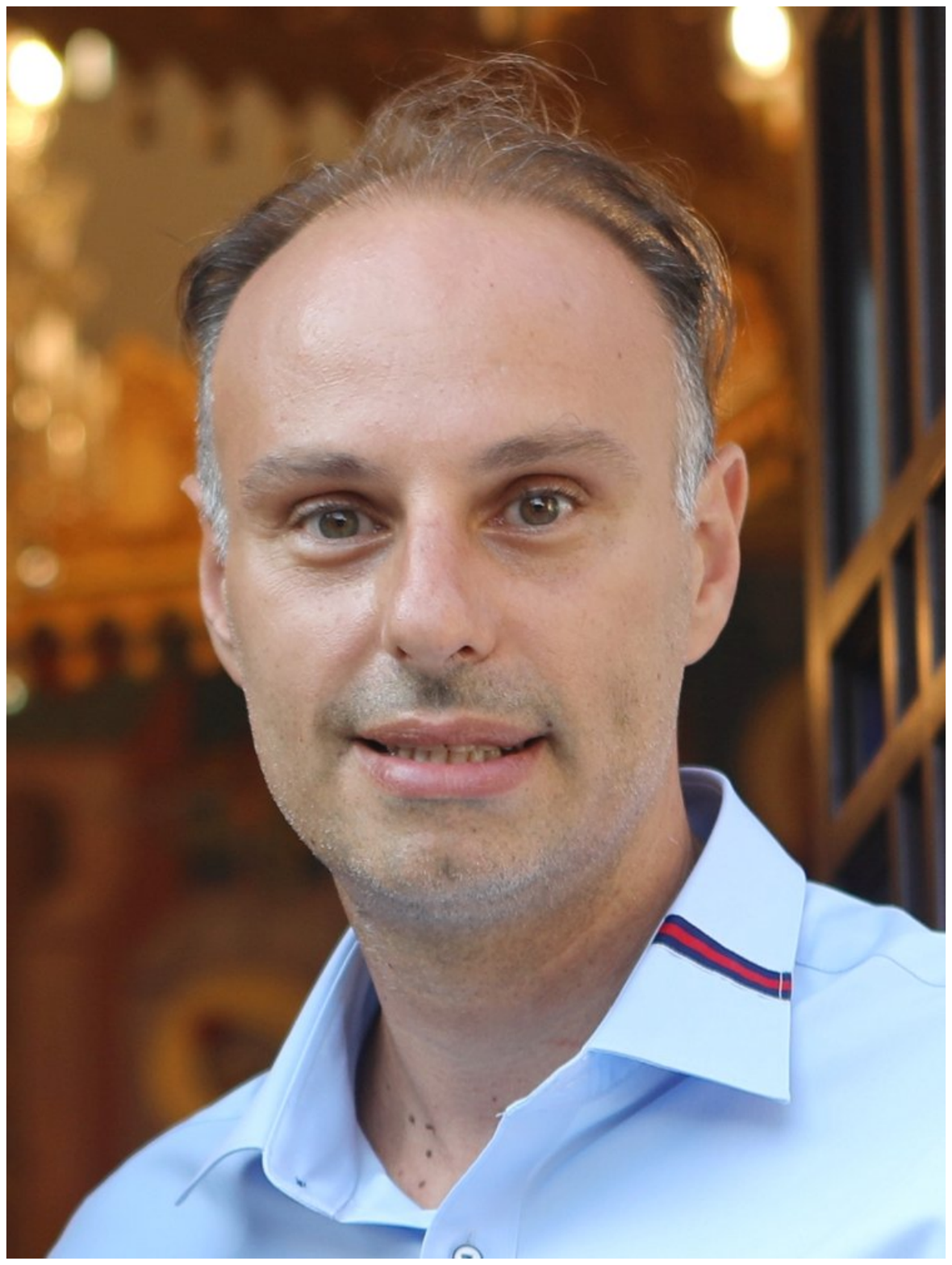}}]{Ioannis Krikidis} (F’19) received the diploma in Computer Engineering from the Computer Engineering and Informatics Department (CEID) of the University of Patras, Greece, in 2000, and the M.Sc and Ph.D degrees from École Nationale Supérieure des Télécommunications (ENST), Paris, France, in 2001 and 2005, respectively, all in Electrical Engineering. From 2006 to 2007 he worked, as a Post-Doctoral researcher, with ENST, Paris, France, and from 2007 to 2010 he was a Research Fellow in the School of Engineering and Electronics at the University of Edinburgh, Edinburgh, UK. He is currently a Professor at the Department of Electrical and Computer Engineering, University of Cyprus, Nicosia, Cyprus. His current research interests include wireless communications, quantum computing, 6G communication systems, wireless powered communications, and intelligent reflecting surfaces. He serves as an Associate Editor for IEEE Transactions on Wireless Communications. He was the recipient of the Young Researcher Award from the Research Promotion Foundation, Cyprus, in 2013, and the recipient of the IEEE ComSoc Best Young Professional Award in Academia, 2016, and IEEE Signal Processing Letters best paper award 2019. He has been recognized by the Web of Science as a Highly Cited Researcher for 2017-2021. He has received the prestigious ERC Consolidator Grant for his work on wireless powered communications.
\end{IEEEbiography}

\end{document}